\documentclass[11pt]{article}
\usepackage{fullpage,amsthm,amssymb,amsmath, amsfonts}
\usepackage[normalem]{ulem}
\usepackage{authblk}
\usepackage{enumerate}
\usepackage{hyperref}
\usepackage{mathabx}
\usepackage{blkarray}
\usepackage{tikz} 
\usetikzlibrary{backgrounds}
\usetikzlibrary{shapes}
\usetikzlibrary{decorations.markings}
\usepackage{subfig}
\usetikzlibrary{snakes}
\usepackage{algorithmicx,algorithm}
\usepackage{algpseudocode}
\hypersetup{
    pdftitle=   {An FPT algorithm and a polynomial kernel for Linear Rankwidth-$1$ Vertex Deletion},
   pdfauthor=  {Mamadou Moustapha Kant\'e, Eun Jung Kim, O-joung Kwon, and Christophe Paul}
}

\usepackage{amsmath}
\newtheorem{theorem}{Theorem}[section]
\newtheorem{definition}{Definition}[section]
\newtheorem{lemma}[theorem]{Lemma}

\newtheorem{proposition}[theorem]{Proposition}

\newtheorem{RULE}{Reduction Rule}

\theoremstyle{remark}

\usepackage[utf8]{inputenc}
\usepackage{enumerate}
\usepackage{hyperref}
\usepackage{amssymb}
\usepackage{blkarray}
\usepackage{subfig}
\usepackage{epsfig}

\newcommand\abs[1]{\lvert #1\rvert}

\usepackage{amsmath}

\newcommand\obn{\Omega_N}
\newcommand\obt{\Omega_T}

\newcommand\rank{\operatorname{rank}}

\newcommand\gr{\mathcal{G}}

\newcommand\cF{\mathcal{F}}
\newcommand\cB{\mathcal{B}}
\newcommand\cR{\mathcal{R}}

\newcommand\house{$F_1$}
\newcommand\gem{$F_2$}
\newcommand\domino{$F_3$}

\newcommand{\LRWD}{\textsc{LRW1-Vertex Deletion} }
\newcommand{\DEL}{LRW1-deletion set}

\newcommand{\YES}{\textsc{Yes}}
\newcommand{\NO}{\textsc{No}}

\begin{document}
\title{An FPT algorithm and a polynomial kernel for Linear Rankwidth-$1$ Vertex Deletion}

\author[1]{Mamadou Moustapha Kant\'e  %
\thanks{ E-mail address: \texttt{mamadou.kante@isima.fr}}}
\affil[1]{LIMOS, CNRS - Clermont Universit\'e, Universit\'e Blaise Pascal, France.}

\author[2]{Eun Jung Kim  %
\thanks{ E-mail address: \texttt{eunjungkim78@gmail.com}}}
\affil[2]{LAMSADE, CNRS - Universit\'e Paris Dauphine, France.}

\author[3]{O-joung Kwon\thanks{Supported by ERC Starting Grant PARAMTIGHT (No. 280152). The work was partially done while at Department of Mathematical Sciences, KAIST, and supported by 
 Basic Science Research
  Program through the National Research Foundation of Korea (NRF)
  funded by  the Ministry of Science, ICT \& Future Planning
  (2011-0011653).}%
\thanks{E-mail address: \texttt{ojoungkwon@gmail.com}}}
\affil[3]{Institute for Computer Science and Control, Hungarian Academy of Sciences, Budapest, Hungary.}

\author[4]{Christophe Paul  
\thanks{Supported by the ``Chercheur d'avenir -- Languedoc-Roussillon'' project KERNEL}
\thanks{E-mail address: \texttt{christophe.paul@lirmm.fr} \\
 An extended abstract appeared in 
  Proc. 10th International Symposium on Parameterized and Exact Computations, 2015~\cite{KanteKKP2015}.}}
\affil[4]{LIRMM, CNRS  - Universit\'e Montpellier, France.}

\date{\today}

\maketitle

\begin{abstract}
\emph{Linear rankwidth} is a linearized variant of rankwidth, introduced by Oum and Seymour [Approximating clique-width and branch-width. \newblock {\em J. Combin. Theory Ser. B}, 96(4):514--528, 2006]. Motivated from recent development on graph modification problems regarding classes of graphs of bounded treewidth or pathwidth, we study the {\sc Linear Rankwidth-$1$ Vertex Deletion} problem (shortly, {\sc LRW1-Vertex Deletion}). In the {\sc LRW1-Vertex Deletion} problem, given an $n$-vertex graph $G$ and a positive integer $k$, we want to decide whether there is a set of at most $k$ vertices whose removal turns $G$ into a graph of linear rankwidth at most $1$ and find such a vertex set if one exists. While the meta-theorem of Courcelle, Makowsky, and Rotics implies that \LRWD  can be solved in time $f(k)\cdot n^3$ for some function $f$,  it is not clear whether this problem allows a running time with a modest exponential function. 

We first establish that \LRWD  can be solved in time $8^k\cdot n^{\mathcal{O}(1)}$. The major obstacle to this end is how to handle a long induced cycle as an obstruction. To fix this issue, we define \emph{necklace graphs} and investigate their structural properties. Later, we reduce the polynomial factor by refining the trivial branching step based on a cliquewidth expression of a graph, and obtain an algorithm that runs in time $2^{\mathcal{O}(k)}\cdot n^4$. We also prove that the running time cannot be improved to $2^{o(k)}\cdot n^{\mathcal{O}(1)}$ under the Exponential Time Hypothesis assumption. Lastly, we show that the \LRWD problem admits  a polynomial kernel.
\end{abstract}

\section{Introduction}

In a parameterized problem, we are given an instance $(x,k)$, where $k$ is a secondary measurement, called as the \emph{parameter}. 
The central question in parameterized complexity is whether a parameterized problem admits an algorithm with running time $f(k)\cdot \abs{x}^{\mathcal{O}(1)}$, called a \emph{fixed parameter tractable algorithm} (shortly, an \emph{FPT algorithm}),
where $f$ is a function depending on the parameter $k$ alone and $\abs{x}$ is the input size.
A parameterized problem admitting such an algorithm is said to be \emph{fixed-parameter tractable}, or {\em FPT} in short. 
As we study a parameterized problem when its unparameterized decision version is NP-hard, the function $f$ is super-polynomial in general. 
For many natural parameterized problems, the function $f$ is overwhelming~\cite{FG04} or even non-explicit~\cite{RS2004}, especially when the algorithm is indicated by a meta-theorem. Therefore, a lot of research effort focus  on designing an FPT algorithm with affordable super-exponential part in the running time. We are especially interested in solving a parameterized problem in {\em single-exponential} FPT time, that is, in time $c^k\cdot n^{O(1)}$ for some constant $c$.

One of techniques to handle parameterized problems is the \emph{kernelization algorithm}.
A kernelization algorithm takes an instance $(x, k)$ and outputs an instance $(x', k')$ in time polynomial in $\abs{x}+k$ satisfying that
(1) $(x,k)$ is a \textsc{Yes}-instance if and only if $(x',k')$ is a \textsc{Yes}-instance,
(2) $k'\le k$, and $\abs{x'}\le g(k)$ for some function $g$.
The reduced instance is called a \emph{kernel} and the function $g$ is called the \emph{size} of the kernel.
A parameterized problem is said to admit a \emph{polynomial kernel} if there is a kernelization algorithm that reduces
the input instance into an instance with size bounded by a polynomial function $g(k)$ in $k$.

Graph modification problems are typically formulated as follows: given an input graph $G$ and a fixed set $\mathsf{O}$ of elementary operations and a graph property $\Pi$, the objective is to transform $G$ into a graph $H\in \Pi$ by applying at most $k$ operations from $\mathsf{O}$. Vertex deletion, edge deletion/addition or contraction are examples of such elementary operations. 

The graph property $\Pi$ having treewidth or pathwidth at most $w$ has received in-depth attention as many problems become tractable on graphs of small treewidth. The celebrated Courcelle's theorem~\cite{Cou90} implies that every graph property expressible in monadic second order logic of the second type ($\sf{MSO}_2$) can be verified in time $f(w)\cdot n$, when the input $n$-vertex graph has treewidth at most $w$. Furthermore, having small treewidth frequently facilitates the design of a dynamic programming algorithm whose running time is much faster than that of the all-round algorithm from the Courcelle's meta-theorem. Therefore, it is reasonable to measure how close an instance is from ``an island of tractability within an ocean of intractable problems''~\cite{GaspersS12}.

In the context of treewidth, the vertex deletion problems for the special cases of $w = 0$ and $w = 1$ correspond to the well-known {\sc Vertex Cover} and {\sc Feedback Vertex Set} problems respectively.
Generally, for fixed $w$, the  {\sc Treewidth-$w$ Vertex Deletion} can be solved in time $f(w, k)\cdot n$ implied by Courcelle's meta-theorem~\cite{Cou90}. As the function $f$ subsumed in the meta theorem is gigantic, it is natural to ask whether the exponential function in the running time can be rendered realistic. Recent endeavor pursuing this question culminated in establishing that for fixed $w$, the {\sc Treewidth-$w$ Vertex Deletion} can be solved in single-exponential FPT time~\cite{FominLMS12,KLPRRSS13}. 

As for pathwidth, \textsc{Pathwidth-$1$ Vertex Deletion} was first studied  by Philip, Raman, Villanger~\cite{PhilipRV2010}, and later Cygan, Pilipczuk, Pilipczuk, Wojtaszczyk~\cite{CyganPPW2012} showed that \textsc{Pathwidth-1 Vertex Deletion} can be solved in time $4.65^k\cdot n^{\mathcal{O}(1)}$ and it admits a quadratic kernel. Using the general method developed for \textsc{Treewidth-$w$ Vertex Deletion}~\cite{FominLMS12,KLPRRSS13}, the \textsc{Pathwidth-$w$ Vertex Deletion} problem also admits a single-exponential FPT algorithm.

\subparagraph{Linear rankwidth.} 
\emph{Rankwidth} was introduced by Oum and Seymour~\cite{OS2004} for efficiently approximating \emph{cliquewidth}. 
Compared to cliquewidth, there are some containment relations, called \emph{vertex-minors} and \emph{pivot-minors}~\cite{Oum05}, 
where the rankwidth of a graph does not increase when taking those relations.
With these relations, 
rankwidth has been intensively studied to extend  results for treewidth and graph minors~\cite{AKK2014, CourcelleK09,GanianH10,JKO2014, Kante2012, Oum05, Oum2006, OS2004}.

\emph{Linear rankwidth} is a linearized variation of rankwidth as pathwidth is the linearized variant of treewidth. While treewidth and pathwidth are small only on sparse graphs, dense graphs may have small rankwidth or linear rankwidth. For instance, complete graphs, complete bipartite graphs, and threshold graphs~\cite{ChvatalH1977} have linear rankwidth at most $1$ even though they  have unbounded treewidth. 

Linear rankwidth is deeply related to \emph{matroid pathwidth}, also known as \emph{trellis-width}, introduced by Kashyap~\cite{Kashyap08}.
Matroid pathwidth has been studied in some matroid theory literature~\cite{GeelenGW2006, HallOS2007, KoutsonasTY2014}. 
Kashyap~\cite{Kashyap08} showed that it is NP-hard to compute the matroid pathwidth of a binary matroid given with its matrix representation.
From the relation between a binary matroid and its fundamental graph due to Oum~\cite{Oum05},  
one can also deduce that it is NP-hard to compute the linear rankwidth of a graph.
Recently, 
Jeong, Kim, and Oum~\cite{JeongKO2016} showed that for fixed $k$, there is a cubic-time algorithm to test whether an input graph has linear rankwidth at most $k$ or not, and 
output such an ordering if one exists.

Ganian~\cite{Ganian10} pointed out that some NP-hard problems, such as computing pathwidth, can be solved in polynomial time on graphs of linear rankwidth at most $1$.
Generally, the meta-theorem by Courcelle, Makowsky, and Rotics~\cite{CourcelleMR00} states that for every graph property $\Pi$ expressible in monadic second order logic of the first type ($\sf{MSO}_1$) and fixed $k$, there is a cubic-time algorithm for testing whether a graph of rankwidth at most $k$ has property $\Pi$.
As rankwidth is always less than or equal to linear rankwidth, 
those problems are tractable on graphs of bounded linear rankwidth as well.

In the same context, it is natural to ask whether there is an FPT algorithm for the \textsc{(Linear) Rankwidth-$w$ Vertex Deletion} problem, that is a problem asking whether for a given graph $G$ and a positive integer $k$, $G$ contains a vertex subset of size at most $k$ whose deletion makes $G$ a graph of (linear) rankwidth at most $w$. 
It is only known that for fixed $w$, both problems are FPT from the meta-theorem on graphs of bounded rankwidth~\cite{CourcelleMR00}. 
We discuss it in more detail in Section~\ref{sec:remark}.
However, as the function of $k$ obtained from the meta-theorem is enormous, it is interesting to know whether there is a single-exponential FPT algorithm for both problems, like \textsc{Treewidth-$w$ Vertex Deletion}. 
Also, to the best of our knowledge, there was no known previous result whether the  \textsc{(Linear) Rankwidth-$w$ Vertex Deletion} problem admits a polynomial kernel for fixed integer $w$.

\subparagraph{Our contributions.}  In this paper, we show that the \textsc{Linear Rankwidth-$1$ Vertex Deletion} problem admits a single-exponential FPT algorithm and a polynomial kernel. This is a first step towards a goal of investigating whether the \textsc{(Linear) Rankwidth-$w$ Vertex Deletion} problem admits a single-exponential FPT algorithm or has a polynomial kernel. 

\smallskip
\noindent
\fbox{\parbox{0.97\textwidth}{
{\sc Linear Rankwidth-$1$ Vertex Deletion} ({\sc LRW1-Vertex Deletion}) \\
\textbf{Input :} A graph $G$, a positive integer $k$ \\
\textbf{Parameter :} $k$ \\
\textbf{Question :} Does $G$ have a vertex subset $S$ of size at most $k$ whose removal makes $G$ a graph of linear rankwidth at most one?} }

\begin{theorem}\label{thm:main1}
The \LRWD problem can be solved in time $8^k\cdot \mathcal{O}(n^{8})$, and also can be solved in time $2^{\mathcal{O}(k)}\cdot n^4$.
\end{theorem}

\begin{theorem}\label{thm:main2}
The \LRWD problem has a kernel with $\mathcal{O}(k^{33})$ vertices.
\end{theorem}

We note that 
several graph classes with a certain path-like structure have been
studied for parameterized vertex deletion 
problems. Such classes include graphs of pathwidth-$1$~\cite{PhilipRV2010, CyganPPW2012},
proper interval graphs~\cite{FSV2012, VV2013}, unit interval graphs~\cite{BevernKM2010,Cao2015}, and interval graphs~\cite{CaoM2015}. A common approach in the previous work was to use the characterization of the structures obtained after removing small obstructions. We also characterize graphs excluding small obstructions for graphs of linear rankwidth at most $1$.

We investigate a new class of graphs, called \emph{necklace graphs}, which are close to graphs of linear rankwidth at most $1$. 
Briefly speaking, necklace graphs, when viewed locally, are graphs of linear rankwidth at most $1$, but they may have long induced cycles.
In Section~\ref{sec:necklacegraph}, we show that every connected graph having no obstructions of size at most $8$ for graphs of linear rankwidth at most $1$ is either a graph of linear rankwidth at most $1$ or a necklace graph (Theorem~\ref{thm:mainlrw}).
Combining a simple branching algorithm and a polynomial-time algorithm to find a minimum deleting set on necklace graphs, we obtain an FPT algorithm for \LRWD with running time $8^k\cdot \mathcal{O}(n^8)$
in the beginning of Section~\ref{sec:fptthreaddel}. 

One might ask whether the polynomial factor $n^8$ can be reduced.
This running time appears as we start with finding obstructions of size at most $8$.
Indeed, we can improve it using a dynamic programming algorithm to find an induced subgraph of fixed size in a graph of bounded cliquewidth.
If the rankwidth of a given graph is more than $k+1$, then the instance is trivially a \NO-instance because rankwidth can be decreased by at most $1$ when removing a vertex.
Using the approximation algorithm due to Oum~\cite{Oum2006}, 
we can decide whether a given graph has rankwidth at most $k+1$ and if so, outputs a rank-decomposition of width at most $3(k+1)+1=3k+4$ and
also a $(2^{3k+5}-1)$-cliquewidth expression, in time $2^{\mathcal{O}(k)}\cdot n^{4}$.
Then we develop a branching algorithm using the cliquewidth expression, and
finally achieve an FPT algorithm with running time $2^{\mathcal{O}(k)}\cdot n^{4}$.
In Section~\ref{sec:lowerbound}, we prove that the running time of our algorithms cannot be reduced to $2^{o(k)}\cdot n^{\mathcal{O}(1)}$ under a reasonable assumption.
\begin{theorem}
There is no $2^{o(k)}\cdot n^{\mathcal{O}(1)}$-time algorithm for \textsc{LRW1 Vertex Deletion}, unless Exponential Time Hypothesis (ETH) fails.
\end{theorem}

In Section~\ref{sec:polykerthreaddel}, we obtain a polynomial kernel for the \LRWD problem. We start with hitting obstructions of size at most $8$ using the Sunflower lemma, and taking a minimum deleting set on the remaining necklace graph. The union of two sets will have size bounded by a polynomial function in $k$, and its removal makes an input graph into a graph of linear rankwidth at most $1$. 
Graphs of linear rankwidth at most $1$ can be seen as graphs obtained by connecting certain blocks, called \emph{thread blocks}, like a path (Theorem~\ref{thm:structurethread}).
The main difficulty for reducing the remaining part is to shrink a large thread block, and 
we can resolve this issue using the set obtained by the Sunflower lemma.
We remark that a similar idea was used by Fomin, Saurabh, and Villanger~\cite{FSV2012} to obtain a polynomial kernel for the {\sc Proper Interval Vertex Deletion} problem.
We conclude the paper with further discussions in Section~\ref{sec:remark}.

\section{Preliminaries}\label{sec:preliminaries}

In this paper, all graphs are finite and undirected, if not mentioned. 
For a graph $G$, we denote by $V(G)$ and $E(G)$ the vertex set and the edge set of a graph $G$, respectively.  
Let $G$ be a graph.
For $x\in V(G)$, let $N_G(x)$ denote the neighborhood of $x$. 
Let $S$ be a subset of $V(G)$. We denote by $G[S]$ the subgraph of $G$ induced on $S$ and we define $G\setminus S:=G[V(G)\setminus S]$.
For short we write $G\setminus x$ instead of $G\setminus \{x\}$ for $x\in V(G)$.  A vertex $v$ of $G$ is called a \emph{pendant vertex} if
$\abs{N_G(v)}=1$. 
The subset of vertices $\partial_{G}(S)\subseteq S$ is the set of all vertices of $S$ that have a neighbor in $V(G)\setminus S$. 

A graph $H$ is an \emph{induced subgraph} of a graph $G$ if $H=G[S]$ for some $S\subseteq V(G)$.
For a set $\mathcal{F}$ of graphs, a graph $G$ is \emph{$\mathcal{F}$-free} if $G$ has no induced subgraph isomorphic to a graph in $\mathcal{F}$. 

 A vertex $v$ of $G$ is called a {\em cut vertex} if the removal of $v$ from $G$ strictly increases the number of connected components. 
 A maximal connected subgraph of a graph without a cut vertex is called a {\em block}.
 Note than an edge can be a block.

The path on the vertex set $\{v_1, \ldots, v_n\}$ and the edge set $\{v_iv_{i+1}:1\le i\le n-1\}$ will be denoted by $v_1v_2 \cdots v_n$, and 
the cycle on the vertex set $\{v_1, \ldots, v_n\}$ and the edge set $\{v_1v_2, \ldots, v_{n-1}v_n, v_nv_1\}$ 
will be denoted by $v_1v_2 \cdots v_nv_1$.
The length of a path is defined as the number of edges in the path. For $n\ge 3$, we denote by $C_n$ the chordless cycle of length $n$. 
A graph $G$ is \emph{distance-hereditary} if for every connected induced subgraph $H$ of $G$ and $v,w\in V(H)$, 
the distance between $v$ and $w$ in $H$ is the same as the distance between $v$ and $w$ in $G$.
For instance the cycle $C=v_1v_2v_3v_4v_5v_1$ of length $5$ is not distance-hereditary, as $v_1$ and $v_3$ have distance $2$ in the graph, 
but they have distance $3$ in $C\setminus v_2$.
A \emph{star}
is a tree with a distinguished vertex adjacent to all other vertices.  A \emph{complete graph} is a graph with all possible edges.

An ordering on a finite set $S$ is a bijective mapping $\sigma:S\to \{1,\ldots, \abs{S}\}$, and we write $x<_{\sigma} y$ if $\sigma(x)<\sigma(y)$, and $\sigma^{-1}$ as the inverse
bijective mapping. 
For an $X\times Y$-matrix $M$ and $X'\subseteq X, Y'\subseteq Y$, let $M[X',Y']$ be the submatrix of $M$ whose rows and columns are indexed by $X'$ and $Y'$, respectively.

When we analyze the running time of an algorithm, we agree that $n=\abs{V(G)}$ and $m=\abs{E(G)}$ if $G$ is an input graph.

\subsection*{Linear rankwidth and thread graphs}
The \emph{adjacency matrix} of a graph $G$, which is a $(0, 1)$-matrix over the binary field, will be denoted by $A_G$.
The \emph{width} of an ordering $\sigma$ of the vertex set of a graph $G$ is 
\[\max_{1\le i\le \abs{V(G)}}\rank (A_G[\{\sigma^{-1}(1), \ldots, \sigma^{-1}(i)\}, V(G)\setminus \{\sigma^{-1}(1), \ldots, \sigma^{-1}(i)\} ]),\] 
where the rank of a matrix is computed over the binary field. The \emph{linear rankwidth} of a graph $G$ is defined as the minimum width over all orderings of $V(G)$.

Ganian~\cite{Ganian10} first characterized graphs of linear rankwidth at most $1$, and he called them \emph{thread graphs}. 
Later, Adler, Farley, and Proskurowski~\cite{AFP2013} gave an easier way to define a thread graph, using a notion of \emph{thread blocks}.
We follow the definition of thread blocks given by Adler, Farley, and Proskurowski, and provide a unified way to define classes of graphs including thread graphs.

A graph $G$ with at least $2$ vertices is called a \emph{thread block}
if there is an ordering $\sigma$ of $V(G)$ and a function $\ell:V(G)\rightarrow \{\{L\}, \{R\}, \{L,R\} \}$ satisfying that
\begin{enumerate}[(1)]
\item $\ell(\sigma^{-1}(1))=\{R\}$ and $\ell(\sigma^{-1}(\abs{V(G)}))=\{L\}$,
\item for $v, w\in V(G)$ with $v<_{\sigma} w$, $vw\in E(G)$ if and only if $R\in \ell(v)$ and $L\in \ell(w)$.
\end{enumerate}
A thread block is a \emph{canonical thread block} if it satisfies on the top of the previous two conditions and the following third condition:
\begin{enumerate}[(3)]
\item $\ell(\sigma^{-1}(2))\ne \{L\}$ if $\abs{V(G)}\neq 2$.
\end{enumerate}
The third condition implies that every canonical thread block has no pendant vertex adjacent to its first vertex $\sigma^{-1}(1)$.
It will guarantee a unique decomposition of a thread graph into thread blocks. 
We say that $\sigma$ and $\ell$ are the \emph{ordering} and \emph{labeling} of $G$ respectively, and 
say that $\sigma^{-1}(1)$ and $\sigma^{-1}(\abs{V(G)})$ are the \emph{first and last} vertices of $G$, respectively.
See Figure~\ref{fig:threadblock} for an example.

\begin{figure}
\centerline{\includegraphics[trim={4cm 11.5cm 6cm 12cm},clip,scale=0.50]{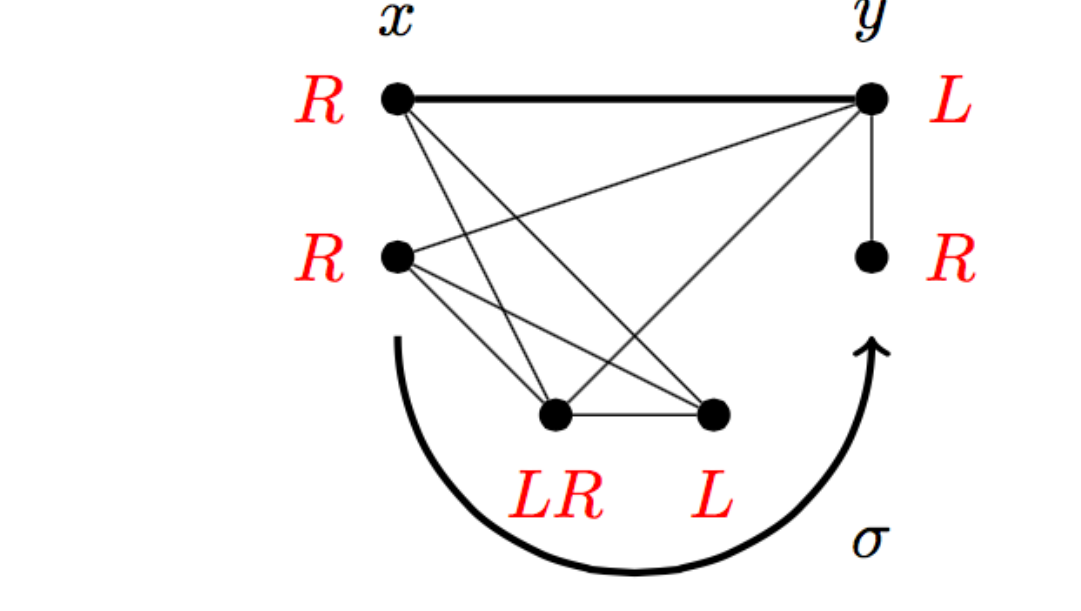}}
\caption{An example of a thread block whose first and last vertices are $x$ and $y$, respectively.} \label{fig:threadblock}
\end{figure}
\begin{figure}
\centerline{\includegraphics[clip,scale=0.7]{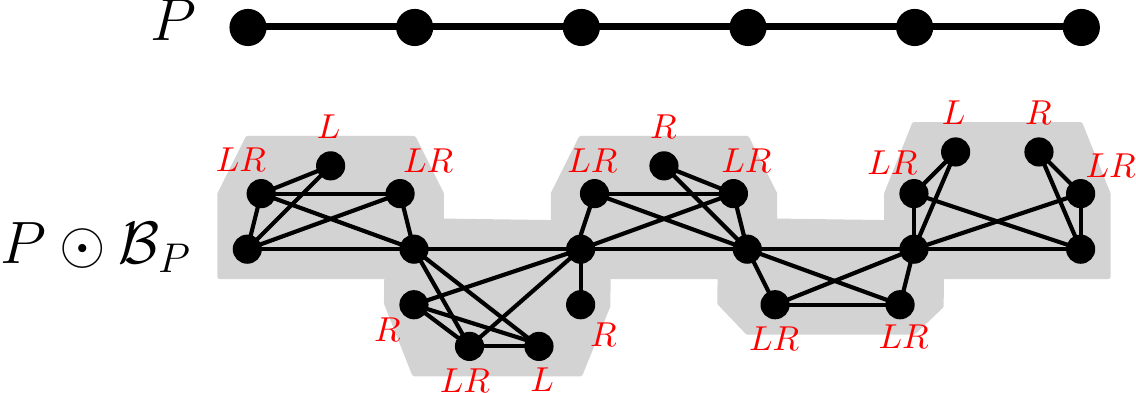}}
\caption{An example of a connected thread graph and its canonical thread decomposition.} \label{fig:threadgraph}
\end{figure}

Let $D=(V_D, A_D)$ be a digraph.
A  set $\{G_{xy}: xy\in A_D\}$ of thread blocks is said to be 
\emph{mergeable with $D$} if
\begin{enumerate}[(1)]
\item for every arc $xy$ of $A_D$, $G_{xy}$ is a thread block whose first and last vertices are $x$ and $y$, respectively, and
\item  for two distinct arcs $x_1y_1, x_2y_2$ of $A_D$, $V(G_{x_1y_1})\cap V(G_{x_2y_2})=\{x_1,y_1\}\cap \{x_2, y_2\}$.
\end{enumerate}
For a digraph $D=(V_D, A_D)$ and a set of thread blocks $\mathcal{B}_D=\{G_{xy} : xy\in A_D\}$ mergeable with $D$, we define $D\odot\mathcal{B}_D$ as the graph $G$ with the vertex set  $V(G)=\bigcup_{xy\in A_D} V(G_{xy})$ and the edge set $E(G)=\bigcup_{xy\in A_D} E(G_{xy})$.
We say that $D$ is the \emph{underlying digraph} of $G$ and that $D\odot \mathcal{B}_D$ is a \emph{thread decomposition} of $G$.
We say that  $D\odot \mathcal{B}_D$ is a \emph{canonical thread decomposition} of $G$
if every thread block of $\mathcal{B}_D$ is a canonical thread block.

\begin{definition}[Thread graph]
A connected graph $G$ is a \emph{thread graph}
if $G$ is either an one vertex graph or $G=P\odot\mathcal{B}_P$ for some directed path $P$ and some set of thread blocks $\mathcal{B}_P$ mergeable with $P$. A graph is a \emph{thread graph} if each of its connected components is a thread graph. 
\end{definition}

See Figure~\ref{fig:threadgraph} for an example of a connected thread graph, and its canonical thread decomposition.
It is not hard to observe that every connected thread graph admits a canonical thread decomposition; one can obtain a canonical thread decomposition from any thread decomposition by rearranging pendant vertices adjacent to a vertex in the underlying digraph.
In Lemma~\ref{lem:splittreetothreadblock}, we will give a polynomial-time algorithm that given a connected thread graph, outputs its canonical thread decomposition.

The following structural properties of thread graphs will be used in later sections.

\begin{lemma}\label{lem:structure}
Let $n\ge 3$ be an integer, and let $G$ be a connected thread graph such that
$G$ admits a canonical thread decomposition $G=P\odot \cB_{P}$ where 
$P=v_1 \cdots v_n$ and
$\cB_{P}:=\{B_i=B_{v_iv_{i+1}}:1\le i\le n-1\}$.
Then the following are satisfied.
\begin{enumerate}[(1)]
\item For every $v\in V(G)\setminus V(P)$, $v$ has a neighbor in $P$. 
\item Vertices $v_2, \ldots, v_{n-1}$ are cut vertices of $G$.
Moreover, every cut vertex of $G$ is contained in $P$.
\item For each $1\le i\le n-1$, $V(B_i)$ is exactly the union of the set of all pendant vertices adjacent to $v_{i+1}$ and the vertex set of the block of $G$ containing $v_{i}, v_{i+1}$.
\end{enumerate}
\end{lemma}
\begin{proof}
(1) Let $B_i$ be the thread block containing $v$. Depending on the label of $v$ in $B_i$, 
$v$ is adjacent to at least one of $v_i, v_{i+1}$.

\medskip
\noindent
(2) The first statement came from the definition of a connected thread graph that 
for $1\le i, j\le n-1$ with $i\neq j$, $V(B_i)\cap V(B_j)=\{v_i, v_{i+1}\}\cap \{v_j, v_{j+1}\}$. 
That is, for each $2\le i\le n-1$, all paths from $v_{i-1}$ to $v_{i+1}$ must pass through $v_i$, and it implies that $v_i$ is a cut vertex of $G$.

If $v\in V(G)\setminus V(P)$, 
then by the statement (1), every vertex of $G\setminus v$ has a neighbor in $P$ and thus $G\setminus v$ is connected.
It implies that every cut vertex of $G$ is contained in $P$.

\medskip
\noindent
(3) Let us fix $i\in \{1, \ldots, n-1\}$ and let $\sigma$ and $\ell$ be the ordering and labeling of $B_i$, respectively.
Let $S\subseteq V(G)$ be the vertex set containing all pendant vertices adjacent to $v_{i+1}$ and the vertex set of the block of $G$ containing $v_{i}$ and $v_{i+1}$.
We need to prove that $S=V(B_i)$.

To prove that $S\subseteq V(B_i)$, let $v\in S\setminus \{v_i, v_{i+1}\}$. 
Observe that by construction of $G$ and the fact that $v\in S$, if $v$ is not a pendant vertex and satisfies $N_G(v)\cap V(P)\subseteq \{v_i, v_{i+1}\}$, then $v\in V(B_i)$. 
So, let us assume that $v$ is not a pendant vertex and has a neighbor in $P$ other than $v_i$ and $v_{i+1}$.
If $v$ is adjacent to some vertex $v_j$ of $P$ with $j>i+1$, then  it contradicts the fact that $v_{i+1}$ separates $v_i$ and $v_j$ in $G$ by (2).
Similarly, if $v$ is adjacent to some vertex $v_j$ of $P$ with $j<i$, then it contradicts the fact that $v_i$ separates $v_j$ and $v_{i+1}$ in $G$ by (2).
Thus, we have $N_G(v)\cap V(P)\subseteq \{v_i, v_{i+1}\}$ and $v\in V(B_i)$.
Finally by construction, as $P\odot \cB_{P}$ is a canonical thread decomposition, the pendant vertices adjacent to $v_{i+1}$ belong to $B_i$. We conclude that $S\subseteq V(B_i)$.

We verify that $V(B_i)\subseteq S$. Let $v\in V(B_i)\setminus \{v_i, v_{i+1}\}$.
If $\ell(v)=\{L, R\}$, then $v$ is contained in $S$.
If $\ell(v)=\{L\}$, then by definition of a canonical thread block, there exists $w<_{\sigma} v$ with $R\in \ell(w)$.
So, $v$ is contained in a cycle of length $4$ together with $v_i, v_{i+1}$, and thus, $v\in S$.
If $\ell(v)=\{R\}$ and there exists $v<_{\sigma} w$ with $L\in \ell(w)$, 
then similarly we have $v\in S$, 
and if there are no such a vertex $w$, then $v$ is a pendant vertex adjacent to $v_{i+1}$, that is also contained in $S$.
\end{proof}

\begin{figure}
\centerline{\includegraphics[scale=0.8]{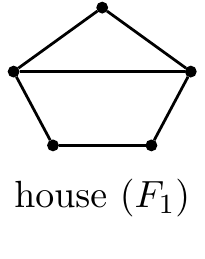} \quad
\includegraphics[scale=0.8]{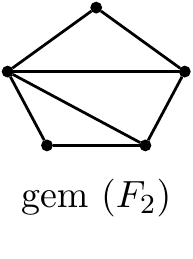} \quad
\includegraphics[scale=0.8]{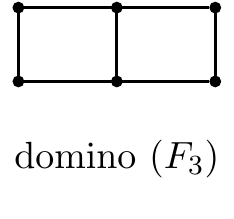} \quad
\includegraphics[scale=0.8]{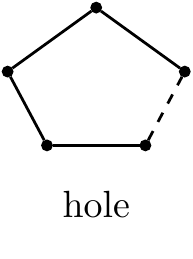}}
\caption{The induced subgraph obstructions for distance-hereditary graphs.}
\label{fig:obsdh}
\end{figure}

\begin{figure}
\centerline{\includegraphics[scale=0.7]{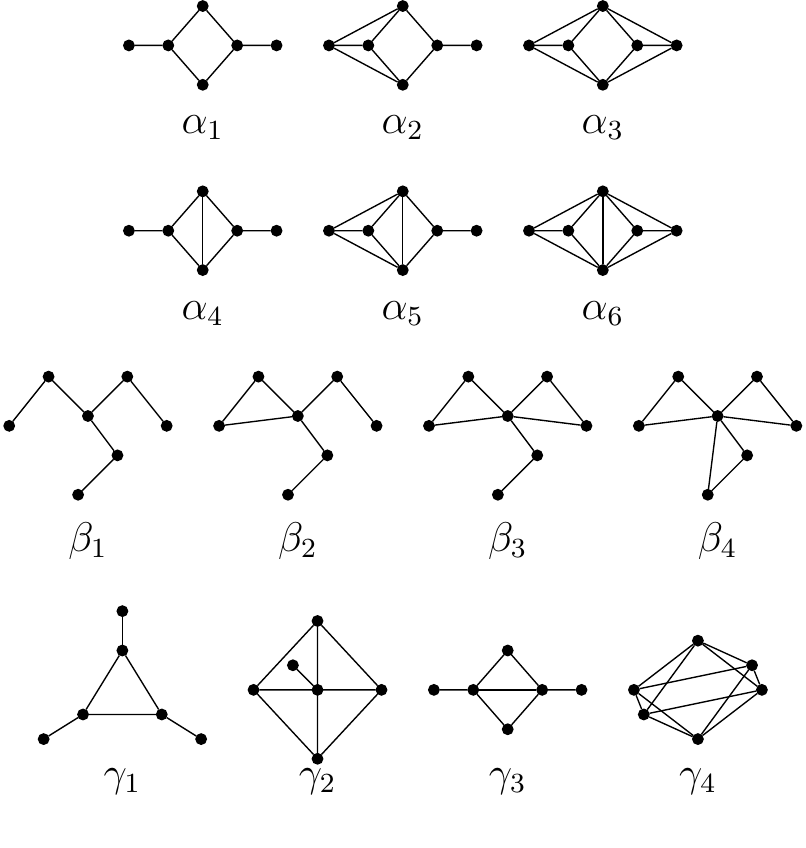}}
\vspace{-0.5cm}
\caption{The induced subgraph obstructions for graphs of linear rankwidth $1$ that are distance-hereditary.}
\label{fig:obslrw1}
\end{figure}

Let $F_1, F_2, F_3$ be the house, gem, domino graphs respectively, which are depicted in Figure~\ref{fig:obsdh}.
The induced subgraph obstructions for graphs of linear rankwidth at most $1$ consist of the set of induced subgraph obstructions for distance-hereditary graphs~\cite{BM1986}, that are \house, \gem, \domino, and induced cycles of length at least $5$, and the set of 14 induced subgraph obstructions for graphs of linear rankwidth at most $1$ that are distance-hereditary, depicted in Figure~\ref{fig:obslrw1}~\cite{AFP2013}.
We define that
\begin{itemize}
\item $\obt$ is the union of $\{\text{\house, \gem, \domino}\}\cup \{C_k : k\ge 5\}$ and the set of $14$ graphs in Figure~\ref{fig:obslrw1}.
\end{itemize}

\begin{theorem}[Ganian~\cite{Ganian10}; Adler, Farley, and Proskurowski~\cite{AFP2013}]\label{thm:structurethread}
For a graph $G$, the following are equivalent.
\begin{itemize}
\item $G$ has linear rankwidth at most $1$.
\item $G$ is a thread graph.
\item $G$ has no induced subgraph isomorphic to a graph in $\obt$.
\end{itemize}
\end{theorem}

We often use the term `thread graphs' for graphs of linear rankwidth at most $1$.
For a graph $G$ and $S\subseteq V(G)$, $S$ is called a \emph{\DEL} if $G\setminus S$ is a graph of linear rankwidth at most $1$.

\subsection*{Obtaining a canonical thread decomposition}

It is known that one can recognize graphs of linear rankwidth at most $1$ in time $\mathcal{O}(\abs{V(G)}+\abs{E(G)})$ using split decompositions~\cite{Bui-XuanKL13, AdlerKK20152}.
Furthermore, we can easily obtain a canonical thread decomposition of a graph of linear rankwidth at most $1$ from its split decomposition, 
but for our knowledge, it was not stated anywhere.
In this subsection, we clarify a procedure to obtain a thread decomposition.
This will be especially used in the kernelization algorithm in Section~\ref{sec:polykerthreaddel}.

We use graph-labelled trees introduced by Gioan and Paul~\cite{GP2012}, which are convenient forms of split decompositions.  
A triple $(T, \cF=\{G_v\}_{v\in V(T)}, \cR=\{\rho_v\}_{v\in V(T)})$ of a tree $T$ and a set of graphs $\cF$ and a set of functions $\cR$ is called a \emph{graph-labelled tree} if 
\begin{itemize}
\item for every node $v$ of degree $k$ in $T$, $G_v$ is a connected graph on $k$ vertices, called \emph{marker vertices},
\item $\rho_v$ is a bijection from the edges of $T$ incident with $v$ to the marker vertices of $G_v$. 
\end{itemize}
Let $v$ be a leaf of $T$.
A node or a leaf $u$ different from $v$ is called \emph{$v$-accessible} if for every edges $xy$ and $yz$ on the path from $u$ to $v$ in $T$,
$\rho_y(xy)$ is adjacent to $\rho_y(yz)$ in $G_y$. 
The \emph{accessibility graph} of a graph-labelled tree $(T, \cF, \cR)$ is the graph $\gr(T, \cF, \cR)$ whose vertex set is the set of all leaves of $T$, 
and $xy\in E(\gr(T, \cF, \cR))$ if and only if $y$ is $x$-accessible. We say that $(T, \cF, \cR)$ is a graph-labelled tree of $\gr(T, \cF, \cR)$.

We give two operations on a graph-labelled tree to define a \emph{reduced graph-labelled tree}.
A \emph{split} of a graph $G$ is a vertex partition $(V_1, V_2)$ of $G$ such that
$\abs{V_1}, \abs{V_2}\ge 2$, and there exist $V_1'\subseteq V_1, V_2'\subseteq V_2$ where the set of edges incident with both $V_1$ and $V_2$ is exactly $\{xy: x\in V_1', y\in V_2'\}$.
 A connected graph is \emph{degenerate} if every vertex partition $(V_1, V_2)$ with $\abs{V_1}, \abs{V_2}\ge 2$ is a split.
It is known that every degenerate graph is either a complete graph or a star graph.
 A node in a graph-labelled tree is called a \emph{clique node} (or a \emph{star node}) if a complete graph (or a star graph, respectively) is assigned to the node.
 A graph without splits is called a \emph{prime graph}.
 
Let $(T, \cF, \cR)$ be a graph-labelled tree of a graph $G$.
Let $v$ be a node of $T$ such that $G_v$ admits a split $(A_1, A_2)$.
For each $i\in \{1,2\}$, 
let $G_i$ be the graph obtained from $G[A_i]$ by adding a new vertex $a_i$ that is adjacent to all vertices in $\partial_{G_v}(A_i)$.
Then the \emph{node-split} operation on the node $v$ with respect to $(A_1, A_2)$
consists of substituting $v$ by two adjacent nodes $v_1, v_2$, respectively labelled by $G_1, G_2$, such that for each $i\in \{1,2\}$, 
\begin{displaymath}
\rho_{v_i}^{-1}(w) = \left\{ \begin{array}{ll}
 v_1v_2 & \textrm{if $w=a_i$,}\\
 \rho_{v}^{-1}(w) & \textrm{if $w\in V(G_i)\setminus \{a_i\}$.}
  \end{array} \right.
\end{displaymath}
The \emph{node-join} operation is the reverse operation of the node-split operation; 
if $vw$ is an edge of $T$, then the \emph{node-join} operation on $vw$ 
consists of contracting $vw$ into a new node $u$ labelled by the graph $G_u$ 
where $G_u$ is the graph obtained from the disjoint union of $G_v$ and $G_w$ 
by deleting $\rho_v(vw)$ and $\rho_w(vw)$ and adding all edges between $N_{G_v}(\rho_v(vw))$ and $N_{G_w}(\rho_w(vw))$.

We say that a star node $v$ is \emph{oriented towards a node $t$} of $T$ if the edge $e$ such that $\rho_v(e)$ is the center
of $G_v$ is on the path in $T$ between $t$ and $v$.

A graph-labelled tree is \emph{reduced} if  
\begin{enumerate}[(1)]
\item every node is either prime or degenerate, and 
\item it contains no edge that connects two degenerate nodes where the node-join operation on this edge results in another degenerate node.
 \end{enumerate}
Cunningham showed the uniqueness of a reduced graph-laballed tree of a connected graph.
Moreover, it can be computed in time $\mathcal{O}(n+m)$.

\begin{theorem}[Cunningham~\cite{Cunningham1982}; Dahlhaus \cite{Dahlhaus00}]\label{thm:cunningham}
Every connected graph $G$ admits a unique reduced graph-labelled tree, and it can be computed in time $\mathcal{O}(n+m)$.
\end{theorem} 
 
For a connected graph $G$, we denote by $ST(G)$ the unique reduced graph-labelled tree obtained from Theorem~\ref{thm:cunningham}.
We call it the \emph{split tree} of $G$.
The following characterization of graphs of linear rankwidth $1$ is crucial.
We give an example of the split tree of a connected thread graph in Figure~\ref{fig:splittree}.

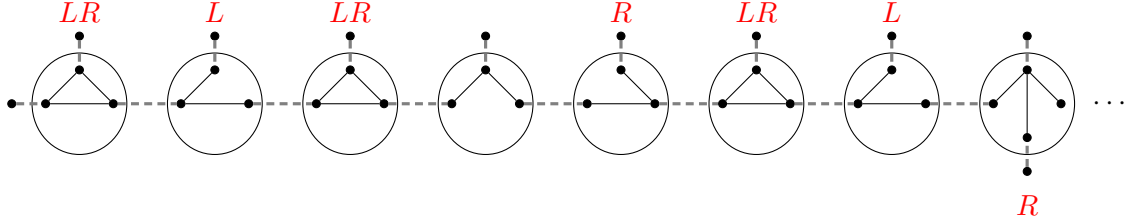
\begin{figure}[t]\centering
\tikzstyle{v}=[circle, draw, solid, fill=black, inner sep=0pt, minimum width=3pt]
 \tikzset{
    photon/.style={decorate,  draw=gray, very thick, , densely dashed},
}

\begin{tikzpicture}[scale=0.045]

\draw (30-40,40) ellipse (14 and 15);
\node(a1) [v] at (30-40,50) {};
\node(a3) [v] at (20-40,40) {};
\node(a4) [v] at (40-40,40) {};
\draw(a1)--(a3)--(a4)--(a1);

\node(a5) [v] at (20-50,40) {};
\node(a6) [v] at (30-40,60) {};
\draw[photon] (a5)--(a3);
\draw[photon](a6)--(a1);

\draw[red] (30-40, 67) node{$LR$};

\draw (30,40) ellipse (14 and 15);

\node(b1) [v] at (30,50) {};
\node(b3) [v] at (20,40) {};
\node(b4) [v] at (40,40) {};
\draw(b1)--(b3)--(b4);

\node(b6) [v] at (30,60) {};
\draw[photon](b6)--(b1);

\draw[red] (30, 67) node{$L$};

\draw (30+40,40) ellipse (14 and 15);

\node(c1) [v] at (30+40,50) {};
\node(c3) [v] at (20+40,40) {};
\node(c4) [v] at (40+40,40) {};
\draw(c1)--(c3)--(c4)--(c1);

\node(c6) [v] at (30+40,60) {};
\draw[photon](c6)--(c1);

\draw[red] (30+40, 67) node{$LR$};

\draw (30+80,40) ellipse (14 and 15);

\node(d1) [v] at (30+80,50) {};
\node(d3) [v] at (20+80,40) {};
\node(d4) [v] at (40+80,40) {};
\draw(d3)--(d1)--(d4);
\node(d6) [v] at (30+80,60) {};
\draw[photon](d6)--(d1);

\draw (30+120,40) ellipse (14 and 15);

\node(e1) [v] at (30+120,50) {};
\node(e3) [v] at (20+120,40) {};
\node(e4) [v] at (40+120,40) {};
\draw(e1)--(e4)--(e3);
\node(e6) [v] at (30+120,60) {};
\draw[photon](e6)--(e1);

\draw[red] (30+120, 67) node{$R$};

\draw (30+160,40) ellipse (14 and 15);

\node(f1) [v] at (30+160,50) {};
\node(f3) [v] at (20+160,40) {};
\node(f4) [v] at (40+160,40) {};
\draw(f1)--(f3)--(f4)--(f1);
\node(f6) [v] at (30+160,60) {};
\draw[photon](f6)--(f1);

\draw[red] (30+160, 67) node{$LR$};

\draw (30+200,40) ellipse (14 and 15);

\node(g1) [v] at (30+200,50) {};
\node(g3) [v] at (20+200,40) {};
\node(g4) [v] at (40+200,40) {};
\draw(g1)--(g3)--(g4);
\node(g6) [v] at (30+200,60) {};
\draw[photon](g6)--(g1);

\draw[red] (30+200, 67) node{$L$};

\draw (30+240,40) ellipse (14 and 15);

\node(h1) [v] at (30+240,50) {};
\node(h2) [v] at (30+240,30) {};
\node(h3) [v] at (20+240,40) {};
\node(h4) [v] at (40+240,40) {};
\draw(h3)--(h1)--(h4);
\draw(h1)--(h2);
\node(h6) [v] at (30+240,60) {};
\node(h7) [v] at (30+240,20) {};
\draw[photon](h6)--(h1);
\draw[photon](h7)--(h2);

\draw[red] (30+240, 10) node{$R$};

\draw[photon] (a4)--(b3);
\draw[photon] (b4)--(c3);
\draw[photon] (c4)--(d3);
\draw[photon] (d4)--(e3);
\draw[photon] (e4)--(f3);
\draw[photon] (f4)--(g3);
\draw[photon] (g4)--(h3);

\draw (30+265, 40) node{$\cdots$};

\end{tikzpicture}
\caption{The canonical split tree of the first two thread blocks of the thread graph in Figure~\ref{fig:threadgraph}.}\label{fig:splittree}
\end{figure}

\begin{theorem}[Bui-Xuan, Kant\'e, and Limouzy~\cite{Bui-XuanKL13}; Adler, Kant\'e, and Kwon~\cite{AdlerKK20152}]\label{thm:charlrw1split}
A connected graph $G$ with the split tree $(T, \cF, \cR)$ is a graph of linear rankwidth at most $1$  if and only if
every graph in $\cF$ is degenerate and 
the tree obtained from $T$ by removing all its leaves is a path.
Thus, one can recognize graphs of linear rankwidth at most $1$ in time $\mathcal{O}(n+m)$.
\end{theorem}

\begin{lemma}\label{lem:splittreetothreadblock}
Given a connected graph $G$ of linear rankwidth at most $1$, 
we can output a canonical thread decomposition $G=P\odot \cB_P$ in time $\mathcal{O}(n+m)$. 
\end{lemma}
\begin{proof}
Using the algorithm in Theorem~\ref{thm:cunningham}, we compute a split tree of $G$ in time $\mathcal{O}(\abs{V(G)}+\abs{E(G)})$.
Let $(T, \{G_v\}_{v\in V(T)}, \{\rho_v\}_{v\in V(T)})$ be the obtained split tree of a connected thread graph $G$.
By Theorem~\ref{thm:charlrw1split}, the tree obtained from $T$ by removing all its leaves is a path and each $G_v$ is degenerate, that is, either a star node or a clique node.
Let $Q$ be the path $v_1v_2 \cdots v_m$ obtained from $T$ by removing all its leaves, and we give a direction on edges so that it is the directed path from $v_1$ to $v_m$.
Note that for each leaf $v$ in $T$ with the neighbor $w$, there exists a unique marker vertex $u$ in $G_w$ such that
$\rho_{w}(vw)=u$.
We denote by $r(v)=u$.

We choose the sequence $v_{i_1}, \ldots, v_{i_t}$ of all star nodes such that  
\begin{itemize}
\item $i_1<i_2< \cdots <i_t$, 
\item for each $1\le j\le t$, there is a leaf $w_j$ of $T$ such that $r(w_j)$ is the center of $G_{v_{i_j}}$.
\end{itemize}
We choose a leaf $w_0\neq w_1$ of $T$ that is adjacent to $v_1$, 
and choose a leaf $w_{t+1}\neq w_t$ of $T$ that is adjacent to $v_m$.
It is not hard to check that $w_0w_1, w_tw_{t+1}\in E(G)$.
Also, from the accessibility among nodes of $T$, one can observe that $w_1, \ldots, w_t$ are cut vertices of $G$, and thus $w_0w_1 \cdots w_tw_{t+1}$ is an induced path of $G$.
Let $P:=w_0w_1 \cdots w_tw_{t+1}$, and we regard it as a directed path from $w_0$ to $w_{t+1}$.

Now, we construct a set of thread blocks $\cB_P$ where 
 $P\odot \cB_P$ is a canonical thread decomposition of $G$.
For each node $v\in V(Q)$, let $\eta(v)$ be the set of all leaves in $T$ that are adjacent to $v$.
For convenience, let $i_0:=0$ and $i_{t+1}:=m$.
For each $0\le j\le t$, we define the following:
\begin{enumerate}
\item Let $B_{j}:=G[\{w_j, w_{j+1}\} \cup \bigcup_{i_j+1 \le \ell \le i_{j+1}  } \eta(v_{\ell})    ]$.
\item We take an ordering $\sigma_j$ of $V(B_j)$ such that $w_j$ and $w_{j+1}$ are the first and last vertices, and 
$v<_{\sigma_j} w$ if there is a directed path from the neighbor of $v$ to the neighbor of $w$ in $Q$. We take an arbitrary ordering for the leaves that have the same neighbor in $T$,  except $w_{j+1}$. 
This can be done in time $\mathcal{O}(\abs{V(G)})$.
\item 
From the definition of $\sigma_i$, 
it is not hard to check that for $v, w\in  V(B_j)\setminus \eta(v_{i_{j+1}})\cup \{w_{j+1}\}$ with $v<_{\sigma_j} w$, $vw\in E(G)$ if and only if 
\begin{itemize}
\item $v=w_j$, or the neighbor of $v$ is either a clique node or a star node oriented towards $v_m$, and 
\item $w=w_{j+1}$, or the neighbor of $w$ is either a clique node or a star node oriented towards $v_1$. 
\end{itemize}
Note that the vertices in $\eta(v_{i_{j+1}})\setminus \{w_{j+1}\}$ are pendant vertices adjacent to $w_{j+1}$ in $G$. 
Following the above observations, we give a labeling $\ell_j$ on $V(B_j)$ such that
for $w\in V(B_j)$ with the neighbor $w'$ in $T$,
 \begin{displaymath}
\ell_j(w) = \left\{ \begin{array}{ll}
\{R\} & \textrm{if $w=w_j$, }\\ 
\{L\} & \textrm{if $w=w_{j+1}$, }\\ 
\{R\} & \textrm{if $w\in  \eta(v_{i_{j+1}})\setminus \{w_{j+1}\}$, }\\
\{L\} & \textrm{if $G_{w'}$ is a star node oriented towards $v_1$, }\\
\{R\} & \textrm{if $G_{w'}$ is a star node oriented towards $v_m$, } \\
\{L,R\} & \textrm{if $G_{w'}$ is a clique node.}
\end{array} \right.
\end{displaymath}
It takes $\mathcal{O}(\abs{V(G)})$ time.
It is not hard to verify that $\sigma_j$ and $\ell_j$ are proper ordering and labeling of $B_j$ using accessibility among nodes of $T$. Moreover,  as $B_{j}$ does not contain a vertex  in $\eta(v_{i_{j}})\setminus \{w_{j}\}$, which is a pendant vertex adjacent to $w_j$ if exists, 
$B_j$ is a canonical thread block whose first and last vertices are $w_j$ and $w_{j+1}$, respectively, with the ordering $\sigma_j$ and the labeling $\ell_j$.

We return $\cB_P:=\{B_0, \ldots, B_t\}$.
\end{enumerate}

Note that $w_1, \ldots, w_t$ are cut vertices of $G$. Thus, it is straightforward to check that
\begin{enumerate}[(1)]
\item $V(B_i)\cap V(B_j)=\{w_i, w_{i+1}\}\cap \{w_j, w_{j+1}\}$ for every pair $i, j$ with $i\neq j$, 
\item  $\bigcup_{0\le i\le t}V(B_i)=V(G)$ and $\bigcup_{0\le i\le t}E(B_i)=E(G)$.
\end{enumerate} 
We conclude that $P\odot \cB_P$ is the canonical thread decomposition of $G$, and it can be computed in time $\mathcal{O}(\abs{V(G)}+\abs{E(G)})$.
\end{proof}

\section{Necklace graphs}\label{sec:necklacegraph}
We generalize the construction of thread graphs from directed paths to directed cycles.

\begin{figure}
\centerline{\includegraphics[trim={0 4cm 0 0}, scale=0.3]{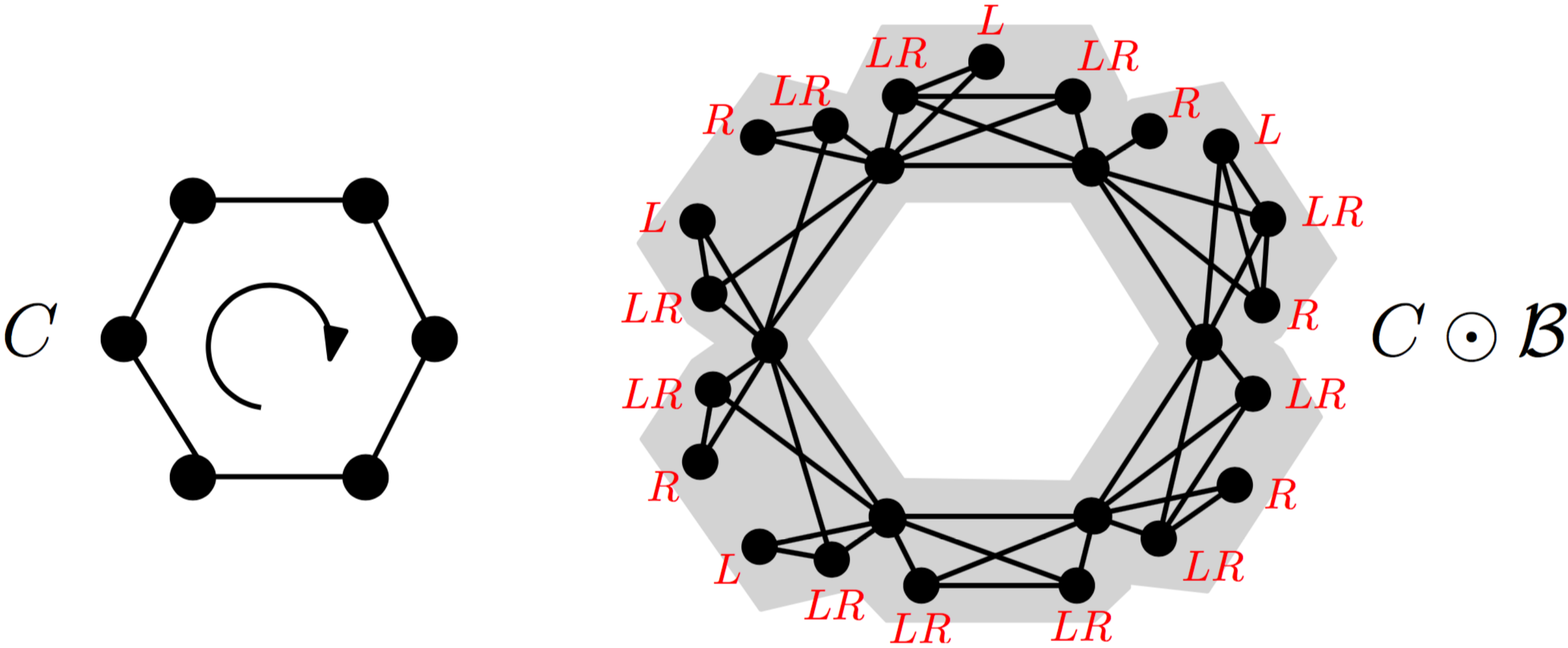}}

\caption{An example of a necklace graph and its canonical thread decomposition.}
\label{fig:necklace}
\end{figure}

\begin{definition}
A connected graph $G$ is called a \emph{necklace graph} if $G=C\odot\mathcal{B}_C$ for some directed cycle $C$ and 
some set of thread blocks $\mathcal{B}_C$ mergeable with $C$. 
\end{definition}

See Figure~\ref{fig:necklace} for an example of a necklace graph.
Let $\obn:=\obt\setminus \{C_h : h\ge 9\}.$ The main result of this section is the following.

\begin{theorem}\label{thm:mainlrw}
Every connected $\obn$-free graph is either a connected thread graph or a necklace graph whose underlying cycle has length at least $9$.
\end{theorem}

Assuming that $G$ is not a  thread graph, to prove Theorem~\ref{thm:mainlrw}, we recursively find an underlying cycle $C$ of length $h\geqslant 9$ and a set of thread blocks $\mathcal{B}$ such that $G=C\odot \mathcal{B}$. At each recursion step, we assume that $G\setminus x=C\odot \mathcal{B}_{G\setminus x}$ and prove that 
\begin{enumerate}[(1)]
\item if $N_G(x)$ contains two vertices of distance at least $3$ in $G$, 
then $G$ contains a graph in $\obn$ or an induced cycle of length $\ell$ with $9\le \ell \le h-1$, 
\item if every pair of vertices of $N_G(x)$ are at distance at most $2$, 
then $N_G(x)$ is contained in the union of two consecutive thread blocks of $\mathcal{B}_{G\setminus x}$, 
and then either $G$ contains a graph in $\obn$ or it is a necklace graph with $C$ as the underlying cycle.
\end{enumerate} 

We need some preliminary lemmas.

\begin{lemma}\label{lem:excludecycle8}
Let $G$ be a necklace graph whose underlying cycle $C$ has length $h\ge 9$. The distance in $G$ between any pair of vertices is at most $h-3$.
\end{lemma}
\begin{proof}
Observe that when merging the (directed) cycle $C$ with a set of thread blocks $\mathcal{B}$ to obtain $G=C\odot \mathcal{B}$, the distance between any two vertices of $C$ in $C$ is the same as the distance in $G$ (viewed as an undirected cycle). As the distance between any pair of vertices in an induced cycle of length $h\ge 9$ is at most $h-5$, the statement follows from the fact that  every vertex in $G$ either belongs to $C$ or has a neighbor in $C$. 
\end{proof}

\begin{lemma}\label{lem:pathtoobs}
Let $h\ge 4$ be an integer.
Let $G$ be a graph and let $v\in V(G)$ such that $G\setminus v$ is an induced path $p_1p_2 \cdots p_h$, and $v$ is adjacent to both $p_1$ and $p_h$ in $G$.
Then $G$ contains an induced subgraph isomorphic to \house, \gem, \domino, or an induced cycle of length at least $5$.
\end{lemma}
\begin{proof}
Let $i_1=1<i_2<\cdots < i_t=h$ be the sequence of integers such that $p_{i_1}, \ldots, p_{i_t}$ are all neighbors of $v$ on $P$.
If $i_{j+1}-i_j\ge 3$ for some $1\le j\le t-1$, then $vp_{i_j}p_{i_j+1} \cdots p_{i_{j+1}}v$ is an induced cycle of length at least $5$.
We may assume that $i_{j+1}-i_j\le 2$ for all $1\le j\le t-1$. If $i_{j+2}-i_{j+1}=i_{j+1}-i_j=2$ for some $1\le j\le  t-2$, then 
$G[\{v, p_{i_j}, p_{i_j+1}, \ldots, p_{i_{j+2}}\}]$ is isomorphic to \domino. If one of $i_{j+2}-i_{j+1}$ and $i_{j+1}-i_j$ is $1$ and the other value is $2$ for some $1\le j\le  t-2$, then  $G[\{v, p_{i_j}, p_{i_j+1}, \ldots, p_{i_{j+2}}\}]$ is isomorphic to \house. So we may assume that $i_{j+1}-i_j=1$ for all $1\le j\le t-1$.
Since $h\ge 4$, $G$ has an induced subgraph isomorphic to \gem.
\end{proof}

\begin{lemma}\label{lem:extendthread}
Let $G$ be a connected thread graph and $v\in V(G)$ such that
$G\setminus v$ admits a canonical thread decomposition $G\setminus v=P\odot \cB_{P}$ where 
$P=v_1v_2v_3v_4v_5$ and
$\cB_{P}:=\{B_i=B_{v_iv_{i+1}}:1\le i\le 4\}$.
If $N_G(v)\subseteq V(B_2)\cup V(B_3)$, then there exists a set $\cB$ of thread blocks mergeable with $P$
such that
$P\odot \cB$ is a canonical thread decomposition of $G$.
\end{lemma}
\begin{proof}
Since $G$ is a connected thread graph,
$G$ admits a canonical thread decomposition $P'\odot \cB_{P'}$ with a directed path $P'$ and a mergeable set of thread blocks $\cB_{P'}$. We first show that $v_2, v_3, v_4$ should be contained in $P'$.
By Lemma~\ref{lem:structure}, it is enough to show that they are cut vertices of $G$.

Note that $v_2$ and $v_4$ are still cut vertices of $G$ as $N_G(v)\subseteq V(B_2)\cup V(B_3)$.
We claim that $v_3$ is also a cut vertex of $G$.
Suppose that $v_3$ is not a cut vertex of $G$. This implies that there is a path from $v_2$ to $v_4$ in $G\setminus v_3$.
We take a shortest path $Q$ among such paths.
If it has length $2$, then $V(Q)\cup \{v_1,v_2, v_3, v_4, v_5\}$ induces a subgraph isomorphic to $\alpha_1$ or $\alpha_4$ depending on the adjacency between $v_3$ and the middle vertex of $Q$.
If it has length at least $3$, then by Lemma~\ref{lem:pathtoobs}, $G[V(Q)\cup \{v_3\}]$ contains an induced subgraph in $\obt$, which contradicts to our assumption.
Thus, $v_3$ is also a cut vertex of $G$.

Since $v_2, v_3, v_4$ are cut vertices of $G$, by (2) of Lemma~\ref{lem:structure}, we have $v_2, v_3, v_4\in V(P')$. 
If $v_2v_3v_4$ is a directed path in $P'$, then 
for each $i\in \{2, 3\}$, we take the subgraph $B_i'$ of $G$ induced on the union of the set of all pendant vertices adjacent to $v_{i+1}$ and the vertex set of the block of $G$ containing $v_{i}, v_{i+1}$.
By (3) of Lemma~\ref{lem:structure}, $B_2'$ and $B_3'$ are the canonical thread blocks listed in $\cB_{P'}$. 
Now, we assume that $v_4v_3v_2$ is a directed path in $P'$.
In this case, for each $i\in \{2,3\}$, 
the subgraph $B_i''$ of $G$ induced on the union of the set of all pendant vertices adjacent to $v_{i}$  and the vertex set of the block of $G$ containing $v_{i}, v_{i+1}$
is a canonical thread block listed in $\cB_{P'}$.
Then we obtain $B_i'$ from $B_i''$ by removing all pendant vertices adjacent to $v_i$ and adding all pendant vertices adjacent to $v_{i+1}$.
It is not hard to observe that $B_i'$ is a canonical thread block whose first and last vertices are $v_i, v_{i+1}$, respectively; we can take a reverse ordering for vertices in the block containing $v_i, v_{i+1}$, and replace each label $\{L\}$ with $\{R\}$ and each label $\{R\}$ with $\{L\}$.

Now, we set $\cB:= \{B_1, B_2', B_3', B_4\}$. We claim that $G=P\odot \cB$. 
As each thread block $B_i$ or $B_i'$ in $\cB$ satisfies that $\partial_{G}(B_i)=\{v_i, v_{i+1}\}$ or $\partial_{G}(B_i')=\{v_i, v_{i+1}\}$, 
we have that 
\begin{itemize}
\item $V(B_1)\cap V(B_2')=\{v_2\}, V(B_2')\cap V(B_3')=\{v_3\}, V(B_3')\cap V(B_4)=\{v_4\}$, and
\item $V(B_1)\cap V(B_3')=V(B_1)\cap V(B_4)=V(B_2')\cap V(B_4)=\emptyset$.
\end{itemize}
As $P\odot \cB_{P}$ is a canonical thread decomposition of $G\setminus v$ and $v\in V(B_2')\cup V(B_3')$, by (3) of Lemma~\ref{lem:structure}, 
every vertex of $G$ is either a pendant vertex adjacent to one of $v_2, \ldots, v_5$, 
or
contained in a block containing two consecutive vertices $v_i, v_{i+1}$.
It implies that $V(G)\subseteq \bigcup \{V(B_1), V(B_2'), V(B_3'), V(B_4)\}$, and we have
\begin{itemize}
\item $V(G)=\bigcup \{V(B_1), V(B_2'), V(B_3'), V(B_4)\}$ and $E(G)=\bigcup \{E(B_1), E(B_2'), E(B_3'), E(B_4)\}$.
\end{itemize}
Therefore, we conclude that $G=P\odot \cB$.
\end{proof}

\begin{proof}[Proof of Theorem~\ref{thm:mainlrw}]
Let $G$ be a connected $\obn$-free graph and suppose that $G$ is not a thread graph. By Theorem~\ref{thm:structurethread}, $G$ contains an induced cycle of length at least $9$. Let $C$ be a shortest cycle among induced cycles of length at least $9$ in $G$. For convenience, let $v_{h+1}:=v_1$. We regard $C$ as a directed cycle on vertex set $\{v_1,\dots, v_{h}\}$ (with $h\ge 9$), where for each $1\le j\le h$, $v_jv_{j+1}$ is an arc. 

We prove by induction on $\abs{V(G)}$ that $G$ is a necklace graph whose underlying cycle is $C$. We may assume that $\abs{V(G)}>\abs{V(C)}$. 
Clearly, $G\setminus v$ is again $\obn$-free graph and
$C$ is a shortest cycle among induced cycles of length at least $9$ in $G\setminus v$.  By the induction hypothesis, 
$G\setminus v$ is a necklace graph whose underlying cycle is $C$ and thereby has a canonical thread decomposition $G\setminus v=C\odot\mathcal{B}_{G\setminus v}$ where $\mathcal{B}_{G\setminus v}=\{B_i=B_{v_iv_{i+1}}:1\le i\le h\}$.

We claim that all vertices of $N_G(v)$ have pairwise distance at most $2$ in $G\setminus v$.
Suppose that there are two vertices $x,y\in N_G(v)$ that have distance at least $3$ in $G\setminus v$. 
By Lemma~\ref{lem:excludecycle8}, the distance in $G\setminus v$ between $x$ and $y$ is at most $h-3$. 
We take a shortest path $Q$ from $x$ to $y$ in $G\setminus v$.
Then by Lemma~\ref{lem:pathtoobs}, 
$G[V(Q)\cup \{v\}]$ contains an induced subgraph isomorphic to either \house, \gem, \domino, or an induced cycle $C_v$ of length $h'$ with $5\le h' \le h-1$.
As $C$ was selected as a smallest induced cycle of length $h\ge 9$, if there is an induced cycle $C_v$, then it has length at most $8$. This contradicts the assumption that $G$ is $\obn$-free.

Thus, all vertices of $N_G(v)$ have pairwise distance at most $2$ in $G\setminus v$. It implies that $N_G(v)\subseteq B_i\cup B_{i+1}$ for some $1\le i\le h$.  Let us consider $H=G[B_{i-1}\cup B_i\cup B_{i+1}\cup B_{i+2}\cup \{v\}]$. Observe that $H\setminus v$ is a thread graph. The facts that $G$ is $\Omega_N$-free and that $v$ does not belong to a cycle of $G$ of length $5$ or more (see paragraph above), imply that $H$ is $\Omega_T$-free and thereby a thread graph (Theorem~\ref{thm:structurethread}). By construction $H\setminus v$ is a thread graph with a canonical thread decomposition $P\odot\cB$ where $\cB=\{B_{i-1}, B_i, B_{i+1}, B_{i+2}\}$ and $P$ is the directed path $v_{i-1}v_i v_{i+1}v_{i+2}v_{i+3}$. So we apply Lemma~\ref{lem:extendthread} on $H$ and $H\setminus v$. Let $P\odot\mathcal{B}'$ be the resulting canonical thread decomposition of $H$. Then $C\odot (B_{G\setminus v}\setminus \cB\cup \cB')$ is a canonical thread decomposition of $G$.
\end{proof}

\section{FPT algorithms for \LRWD}\label{sec:fptthreaddel}

In this section, we give two single-exponential FPT algorithms for the \LRWD problem.
Both FPT algorithms are based on branching algorithms that reduce an input graph to a $\obn$-free graph.
We first show that one can easily find a minimum \DEL\ in a $\obn$-free graph.
It implies that there is an FPT algorithm with running time $8^k\cdot \mathcal{O}(n^8)$ with a simple branching algorithm.
Next, we provide a way to reduce the polynomial factor $n^8$ into $n^4$ using an algorithm to find a fixed subgraph on graphs of bounded cliquewidth.

We prove the following.
\begin{proposition}\label{prop:nltothread2}
Given a $\obn$-free graph $G$, we can compute a minimum LRW1-deletion set of $G$ in time $\mathcal{O}(n(n+m))$.
\end{proposition}

We observe that every necklace graph can be turned into a thread graph by removing a vertex on the underlying cycle.

\begin{lemma} \label{lem:nltothread}
Let $G$ be a connected necklace graph with the underlying directed cycle $C$. For each $v\in V(C)$, $G\setminus v$ is a thread graph.
\end{lemma}
\begin{proof}
Let $C$ be the directed cycle  on vertex set $\{v_1,\dots, v_{h}\}$, where for each $1\le j\le h$, $v_jv_{j+1}$ is an arc. 
and let $\cB_C=\{B_i : v_iv_{i+1}\in E(C)\}$ be the set of thread blocks mergeable with $C$ such that
$G=C\odot \cB_C$.
Because of the symmetry, it is sufficient to show that $G\setminus v_1$ is a thread graph. 
As the pendent vertices of $G$ that are adjacent to $v_1$ become isolated vertices after removing $v_1$, 
we may assume that $G$ has no pendent vertices adjacent to $v_1$. 
For each $1\le i\le h$, let $\sigma_i$ and $\ell_i$ be the ordering and labeling of $G_i$, respectively.

Suppose that $V(B_h)\setminus \{v_1\}\neq \{v_{h}\}$. 
Since $B_h$ has no pendant vertices adjacent to $v_1$, 
the vertex $\sigma^{-1}(\abs{V(B_h)}-1)$ is labelled either $\{L,R\}$ or $\{L\}$.
Let $\sigma_h'$ be the restriction of $\sigma_h$ on $V(B_h)\setminus \{v_1\}$, and
for each $x\in V(B_h)\setminus \{v_1\}$, let 
\begin{displaymath}
\ell_{h}'(x) = \left\{ \begin{array}{ll}
\{L\} & \textrm{if $x=\sigma^{-1}_h(\abs{V(B_h)}-1)$}\\
\ell_{h}(x)  & \textrm{otherwise.}
\end{array} \right.
\end{displaymath}
It is easy to check that $B_h\setminus v_1$ 
is a thread block with the ordering $\sigma_h'$ and labeling $\ell_{h}'$.
If $V(B_h)\setminus \{v_1\}= \{v_{h}\}$, 
then we can regard $B_{h-1}$ as the last thread block of $G\setminus v_1$.

Similarly, 
if $V(B_1)\setminus \{v_1\}=\{v_{2}\}$, 
then we can regard $B_2$ as the first thread block of $G\setminus v_1$.
Otherwise, we regard $B_1\setminus v_1$ as the first thread block.

Let $y:=\sigma^{-1}_1(2)$ and $z:= \sigma^{-1}_h(\abs{V(B_h)}-1)$.
We conclude that
$G\setminus v_1$ is a thread graph with the underlying directed path $P$ 
where
\begin{displaymath}
P = \left\{ \begin{array}{ll}
v_2v_3 \cdots v_{h}  \, \,
\qquad \textrm{if $V(B_h)\setminus \{v_1\}= \{v_{h}\}$ and $V(B_1)\setminus \{v_1\}= \{v_{2}\}$,}\\
yv_2v_3 \cdots v_{h}  
\qquad \textrm{if $V(B_h)\setminus \{v_1\}= \{v_{h}\}$ and $V(B_1)\setminus \{v_1\}\neq \{v_{2}\}$,}\\
v_2v_3 \cdots v_{h}z  
\qquad \textrm{if $V(B_h)\setminus \{v_1\}\neq \{v_{h}\}$ and $V(B_1)\setminus \{v_1\}= \{v_{2}\}$,}\\
yv_2v_3 \cdots v_{h}z  \qquad \textrm{otherwise.} 
\end{array}\right. \qedhere
\end{displaymath}
\end{proof}

\begin{proof}[Proof of Proposition~\ref{prop:nltothread2}]
Let $G$ be a $\obn$-free graph and let $k$ be the minimum size of a LRW1-deletion set of $G$.
By Theorem~\ref{thm:mainlrw},  each connected component of $G$ is either a thread graph or a necklace graph.
For each component $H$ of $G$, we can test whether $H$ is a thread graph or not in time $\mathcal{O}(\abs{V(H)}+\abs{E(H)})$ using Theorem~\ref{thm:charlrw1split}. 
By Lemma~\ref{lem:nltothread}, it is enough to remove exactly one vertex to make  each necklace component a thread graph.
Thus, $k$ is equal to the number of its necklace components. 
Moreover, in each necklace component $H$, we can identify a vertex $v$ on the underlying cycle by testing whether $H\setminus v$ is a thread graph for every vertex $v$ in $H$.  
It takes a $\mathcal{O}(\abs{V(H)} (\abs{V(H)}+\abs{E(H)}))$ time.
Therefore, we can find a minimum LRW1-deletion set of $G$ in time $\mathcal{O}(\abs{V(G)}(\abs{V(G)}+\abs{E(G)}))$.
\end{proof}

We give an FPT algorithm using a simple branching algorithm.

\begin{theorem}\label{thm:algorithm1}
The \LRWD problem can be solved in time $8^k\cdot \mathcal{O}(n^{8})$.
\end{theorem}

\begin{proof}
Let $(G, k)$ be an instance of the \LRWD problem where $G$ is a graph on $n$ vertices. 
First recursively find an induced subgraph of $G$ isomorphic to a graph in $\obn$ and branch by removing one of the vertices in the subgraph.
Because the maximum size of graphs in $\obn$ is $8$,
we can find such a vertex subset in time $\mathcal{O}(n^8)$ if exists.
In the end,
we transform the instance $(G, k)$ into at most $8^k$ sub-instances $(G', k')$ such that each sub-instance consists of a $\obn$-free graph $G'$ and a remaining budget $k'$. 
Thus, this branching step takes a time $8^k\cdot \mathcal{O}(n^8)$. 
Clearly, $(G,k)$ is a \textsc{Yes}-instance if and only if one of sub-instances $(G', k')$ is a \textsc{Yes}-instance.

Let $(G', k')$ be a sub-instance obtained from the branching algorithm.
Since $G'$ is $\obn$-free, 
using the algorithm in Proposition~\ref{prop:nltothread2}, 
we can compute a minimum LRW1-deletion set of $G'$ and decide whether $(G', k')$ is a \textsc{Yes}-instance in time $\mathcal{O}(\abs{V(G')}(\abs{V(G')}+\abs{E(G')}))$.
By checking all sub-instances, we can decide whether $(G, k)$ is a \textsc{Yes}-instance in time $8^k\cdot \mathcal{O}(n^8)$. 
\end{proof}
Theorem~\ref{thm:algorithm1} already gives a single-exponential FPT algorithm for the \LRWD problem, but the polynomial factor $n^8$ makes it impractical.
In the next subsection, we give an algorithm with better polynomial factor, using the branching algorithm based on a cliquewidth expression.

\subsection{Improving the polynomial factor}

The \emph{cliquewidth} of a graph $G$ is the minimum number of labels needed to construct $G$ using the following four operations:
\begin{enumerate}[(1)]
\item Creation of a new vertex $v$ with label $i$ (denoted by $i(v)$).
\item Disjoint union of two labeled graphs $G$ and $H$ (denoted by $G \oplus H$).
\item Joining by an edge each vertex with label $i$ to each vertex with label $j$ ($i\neq j$, denoted by $\eta_{i,j}$). 
\item Renaming label $i$ to $j$ (denoted by $\rho_{i\rightarrow j}$).
\end{enumerate}
Every graph can be defined by an algebraic expression using these four operations. 
Such an expression is called a \emph{$k$-expression} if it uses at most $k$ different labels. Thus, the cliquewidth of $G$ is the minimum $k$ for which there exists a $k$-expression defining $G$.

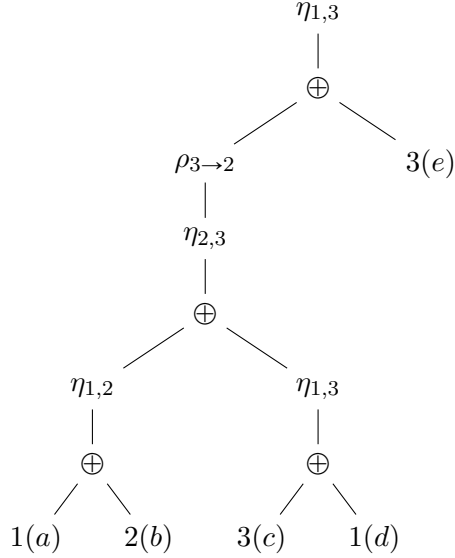
\begin{figure}\center
\begin{tikzpicture}[level distance=10mm]
\tikzstyle{level 1}=[sibling distance=30mm] 
\tikzstyle{level 6}=[sibling distance=15mm] 
\node {$\eta_{1,3}$}
   child {node {$\oplus$} 
    child {node {$\rho_{3\rightarrow 2}$}
      child {node {$\eta_{2,3}$}
      child {node {$\oplus$}
	child {node {$ \eta_{1,2}$}
		child {node {$\oplus$}
		 	child {node {$1(a)$}}
			child {node {$2(b)$}}}}
        child {node {$ \eta_{1,3}$}
			child {node {$\oplus$}
		 	child {node {$3(c)$}}
			child {node {$1(d)$}}}	
			}}
		}}
    child {node {$3(e)$}}};
\end{tikzpicture}
\caption{$3$-expression of $C_5$.} 	\label{fig:c5expression}
\end{figure}

For instance, the cycle $abcdea$ of length $5$ admits the following $3$-expression:
\[  \eta_{1,3}( \rho_{3\rightarrow 2} \eta_{2,3}( \eta_{1,2} (1(a) \oplus 2(b)) \oplus
\eta_{1,3} (3(c) \oplus 1(d)) )  \oplus 3(e))    .\]
We can represent this expression as a tree-structure, depicted in Figure~\ref{fig:c5expression}.
We call this tree the \emph{labelled tree induced by the $k$-expression of $G$}.

As we need the relation between rankwidth and linear rankwidth, we define rankwidth.
A tree is \emph{subcubic} if it has at least two vertices and every inner vertex has degree~$3$. A pair $(T, L)$ is called a \emph{rank-decomposition} of a graph $G$ if $T$ is a subcubic tree and $L$ is a bijection from the vertices of $G$ to the leaves of $T$. For each edge $e$ in $T$, $T\setminus e$ induces a partition $(S_{e} ,T_{e} )$ of the leaves of $T$. The \emph{width} of an edge $e$ is defined as $\rank (A_G[S_{e}, T_{e}])$. The \emph{width} of a rank-decomposition $(T,L)$ is the maximum width over all edges of $T$. The \emph{rankwidth} of $G$ is the minimum width over all rank-decompositions of $G$. If $\abs{V(G)}\leq 1$, then $G$ admits no rank-decomposition and it has rankwidth $0$.

The following approximation algorithm was given by Oum~\cite{Oum2006, OS2004}. 

\begin{theorem}[Oum~\cite{Oum2006, OS2004}]\label{thm:approxcw}
Given a graph $G$ and positive integer $k$, 
one can output a rank-decomposition of width at most $3k+1$ of $G$ or confirm that the rankwidth of $G$ is larger than $k$ in time $2^{\mathcal{O}(k)}\cdot n^4$.
Moreover, given a rank-decomposition of a graph $G$ of width $k$, 
one can output a $(2^{k+1}-1)$-expression of $G$ in time $2^{\mathcal{O}(k)}\cdot n^2$.
\end{theorem}
Remark that the rankwidth of a graph can be decreased by at most $1$ when removing a vertex, 
and the rankwidth of a graph is at most its linear rankwidth.
Thus, if a given graph has rankwidth larger than $k+1$, then we cannot make it a graph of linear rankwidth at most $1$ by removing at most $k$ vertices.
Therefore, using the algorithm in Theorem~\ref{thm:approxcw}, we can assume that the instance is given with a $(2^{3k+5}-1)$-expression, and we can use it to find an induced subgraph of constant size. 

\begin{proposition}\label{prop:inducedoncw}
Let $H$ be a graph on $h$ vertices.
Given a graph $G$ with its $k$-expression, 
one can test whether $G$ contains an induced subgraph isomorphic to $H$ in time $2^{\mathcal{O}(h^2+h\log k)}\cdot n$, 
and output a vertex set inducing $H$ if exists, in the same time complexity.
\end{proposition}
\begin{proof}
Let $v_1, \ldots, v_h$ be the vertices of $H$ and $e_1, \ldots, e_f$ be the edges of $H$.
Let $\phi$ be the given $k$-expression defining $G$, and  
let $T$ be the labelled rooted tree induced by $\phi$.
For every node $t$ of $T$, let $G_t$ be the graph defined at node $t$, and for each $i\in \{1, \ldots, k\}$, let $G_t[i]$ be the subgraph of $G_t$ induced on the vertices with label $i$.

For every node $t$ and every vector $(a_1, \ldots, a_{h}, b_{1}, \ldots, b_{f}) \in \{0, 1, \ldots, k\}^{h} \times \{0, 1\}^{f}$, 
we define the following value:

\begin{itemize}
\item $c[t, (a_1, \ldots, a_{h}, b_{1}, \ldots, b_{f})] = 1$
   if there exists an injective mapping $\eta$ from $V=\{v_i\in V(H):a_i\neq 0\}$ to $V(G_t)$ such that      
   \begin{itemize}
      \item  for each $v_i\in V$, $\eta(v)\in V(G_t[a_i])$,
	\item for each $e_i=v_{\ell} v_m\in E(H)$ with $b_i=1$, we have $a_{\ell}\neq 0$, $a_m\neq 0$, and $\eta(v_{\ell})\eta(v_m)\in E(G_t)$,
	\item for each $e_i=v_{\ell} v_m\in E(H)$ with $b_i=0$, $a_{\ell}\neq 0$ and $a_m\neq 0$, we have $\eta(v_{\ell})\eta(v_m)\notin E(G_t)$,
	\end{itemize}
\item $c[t, (a_1, \ldots, a_{h}, b_{1}, \ldots, b_{f})] = 0$, otherwise.
\end{itemize}
The values $c[t, \cdot]$ will capture all possible inequivalent subgraphs of $H$ at $G_t$, where two subgraphs are equivalent if their corresponding vertices are in the same labels.
A vector $(a_1, \ldots, a_h, b_1, \ldots,b_f)$ is \emph{complete} if for each $1\le i\le h$ and $1\le j\le f$, 
$a_i$ and $b_j$ are non-zero. 
One can observe that $G$ contains an induced subgraph isomorphic to $H$ if and only if 
there is a complete vector $v$ such that $c[r, v]=1$.
A vector $(a_1', \ldots, a_{h}', b_{1}', \ldots, b_{f}')$ is called a \emph{sub-vector} of  $(a_1, \ldots, a_{h}, b_{1}, \ldots, b_{f})$
if for each $1\le x\le h$ and $1\le y\le f$, $a_x'\in \{0, a_x\}$ and $b_y'\in \{0, b_y\}$.

Now, we present how the values of $c[\cdot, \cdot]$ are computed.
At each node $t$, we compute the value $c[t, (a_1, \ldots, a_{h}, b_{1}, \ldots, b_{f})]$ as follows.

\begin{enumerate}
\item (Creation of a new vertex $w$ with label $i$)
Each vertex of $H$ can be mapped to $w$. 
Thus, for each vector $v=(a_1, \ldots, a_{h}, b_{1}, \ldots, b_{f})$ where exactly one of $a_1, \ldots, a_h$ is $i$ and $b_y=0$ for all $y\in \{1, \ldots, f\}$, 
we assign $c[t, v]:=1$ and for all other vectors, assign $c[t, v]:=0$.
\item (Disjoint union node with two children $t_1, t_2$) 
	We take all possible two sub-vectors $v_1=(a^1_1, \ldots, a^1_{h}, b^1_{1}, \ldots, b^1_{f})$ and $v_2=(a^2_1, \ldots, a^2_{h}, b^2_{1}, \ldots, b^2_{f})$ of  $(a_1, \ldots, a_{h}, b_{1}, \ldots, b_{f})$
	such that 
	\begin{itemize}
	\item for each $1\le x\le h$, 
	if $a_x$ is non-zero, then exactly one of $a^1_x$ and $a^2_x$ is equal to $a_x$, and 
	the same thing holds for  $b^1_y$, $b^2_y$, and $b_y$. 
	\end{itemize}	
We assign $c[t, v]:=1$ if there exist such $v_1, v_2$ where $c[t_1, v_1]=c[t_2, v_2]=1$, 
and $c[t, v]:=0$ otherwise.
	This formula holds because 
	$G_t$ contains a subgraph $H'$ if and only if there is a vertex partition $S_1$ and $S_2$ of $H'$ with no edges between $S_1$ and $S_2$ such that
	$H'[S_1]$ is a subgraph of $G_{t_1}$ and $H'[S_2]$ is a subgraph of $G_{t_2}$ ($S_1$ or $S_2$ may be an empty set). 
	This also implies that for a vector $v$ with $b_y=1$ and $e_y=v_{\ell}v_m$, it suffices to consider a pair $(v_1,v_2)$ which satisfies either of the following: 
	$b^1_y=1,a^1_{\ell}\neq 0, a^1_m\neq 0, b^2_y=a^2_{\ell}=a^2_m=0$, or $ b^1_y=a^1_{\ell}=a^1_m=0, b^2_y=1,a^2_{\ell}\neq 0, a^2_m\neq 0$.

	Note that there are at most $2^h\cdot 2^f$ possible pair of vectors.
	Thus, we need $2^h\cdot 2^f$ iterations to compute $c[t,v]$ for fixed $v$.
	In total, we need $(2^h\cdot 2^f)\cdot (k^h\cdot 2^f)=2^{h+2f+h\log k}$ iterations to compute $c[t,v]$ for all $v$. 
\item (Join node with the child $t'$ such that each vertex with label $i$ is joined to each vertex with label $j$)
      If there are $\ell, m\in \{1, \ldots, h\}$ and $y\in \{1, \ldots, f\}$ where $e_y=v_{\ell}v_m$ with $a_{\ell}=i$, $a_m=j$ and 
      $b_y=0$, then we assign $c[t,v]:=0$. This correctly assigns the value as two vertices with labels $i$ and $j$ should be adjacent at this join node.
      We can check it in time $\mathcal{O}(h^2)$.
      Now, we assume that the given vector $v$ satisfies that for all $\ell, m\in \{1, \ldots, h\}$ and $y\in \{1, \ldots, f\}$ with $e_y=v_{\ell}v_m$ with $a_{\ell}=i$ and $a_m=j$,
      we have  $b_y=1$.
      
      We take all sub-vectors $v'=(a_1, \ldots, a_{h}, b_{1}', \ldots, b_{f}')$ of  $(a_1, \ldots, a_{h}, b_{1}, \ldots, b_{f})$
	such that for $y\in \{1, \ldots, f\}$,   
	\begin{itemize}
	\item $b_y'\in \{0, 1\}$ if $e_y=v_{\ell}v_m$ with $a_{\ell}=i$,  $a_m=j$, and $b_y=1$, and 
	\item $b_y'=b_y$ otherwise.
	\end{itemize} 
Then we assign $c[t, v]:=1$ if there exists such $v'$ where $c[t', v']=1$, 
and $c[t, v]:=0$ otherwise.
This formula holds because 
	$G_t$ contains a subgraph $H'$ if and only if $G_{t'}$ contains a subgraph $H''$
	where
	$H'$ is obtained from $H''$ by adding edges between $V(G_{t'}[i])$ and $V(G_{t'}[j])$.
		
   Note that there are at most $2^{f}$ possible sub-vectors.
Thus, we need $2^f\cdot (k^h\cdot 2^f)=2^{2f+h\log k}$ iterations to compute $c[t,v]$ for all $v$. 
      
\item  (Renaming label $i$ to $j$)
	If there is $x\in \{1, \ldots, h\}$ with $a_x=i$, then we assign $c[t,v]:=0$. This correctly assigns the value as there is no vertex with label $i$ at this node.
	 We can check it in time $\mathcal{O}(h)$.
    We assume that the given vector $v$ satisfies that 
    for all $x\in \{1, \ldots, h\}$, $a_x\neq i$.

For each vector $(a_1, \ldots, a_{h}, b_{1}, \ldots, b_{f})$, 
     we take all sub-vectors $v'=(a_1', \ldots, a_{h}', b_{1}, \ldots, b_{f})$ of  $(a_1, \ldots, a_{h}, b_{1}, \ldots, b_{f})$
	such that for $x\in \{1, \ldots, h\}$, 
	\begin{itemize}
	\item $a_x'\in \{i, j\}$ if $a_x=j$, and 
	\item $a_x'=a_x$ otherwise.
	\end{itemize}
We assign $c[t, v]:=1$ if there exists such $v'$ where $c[t', v']=1$, 
and $c[t, v]:=0$ otherwise.
This formula trivially holds as we just change the labels on vertices.

   Note that there are $2^{h}$ possible vectors.
Thus, we need $2^h\cdot (k^h\cdot 2^f)=2^{h+f+h\log k}$ iterations to compute $c[t,v]$ for all $v$. 
\end{enumerate}

Finally, we can test whether $G$ contains an induced subgraph isomorphic to $H$ by checking 
whether there is a complete vector $v$ such that $c[r, v]=1$.
Thus, we can find an induced subgraph of $G$ isomorphic to $H$ in time $2^{\mathcal{O}(h^2+h\log k)}\cdot n$, if exists.
Also, if we find such a vector $v$, then we can track which vertices are contributed to make the vector $v$, 
and thus, we can compute a vertex set inducing $H$ in the same time complexity.
\end{proof}

 Now, we give a second FPT algorithm for the \LRWD problem.
\begin{theorem}\label{thm:algorithm2}
The \LRWD problem can be solved in $2^{\mathcal{O}(k)}\cdot n^4$ time.
\end{theorem}

\begin{proof}
All procedures in the algorithm will be same as in Theorem~\ref{thm:algorithm1} except the refined branching algorithm. 
Let $(G, k)$ be an instance where $G$ is a graph on $n$ vertices.
Using the algorithm in Theorem~\ref{thm:approxcw}, we can decide whether the rankwidth of $G$ is at most $k+1$ in time $2^{\mathcal{O}(k)}\cdot n^4$.
If the rankwidth of the input graph is more than $k+1$, then we cannot make it a graph of linear rankwidth at most $1$ by removing at most $k$ vertices.
Thus, we say that it is a \NO-instance in this case.
Otherwise, the algorithm outputs a $(2^{3k+5}-1)$-expression $\phi$ of $G$.
Let $T$ be the labelled rooted tree induced by $\phi$.
For convenience let $k':=2^{3k+5}-1$.

Let us fix a graph $H$ in $\obn$.
Note that $\abs{V(H)}\le  8$.
Using the algorithm in Proposition~\ref{prop:inducedoncw}, 
we can test whether $G$ contains an induced subgraph isomorphic to $H$ and outputs a vertex set inducing $H$ if exists, in time $2^{\mathcal{O}(\log k')}\cdot n=2^{\mathcal{O}(k)}\cdot n$.
If there is such a vertex set $S$, then 
branch by removing one of the vertices $v$ in $S$ and decrease $k$ by $1$.
Note that we can obtain the cliquewidth expression for $(G\setminus v, k-1)$ from $\phi$ just by removing the node introducing $v$.
We recurse this branching algorithm until there are no such vertex sets $S$.

Combining the remaining steps described in the proof of Theorem~\ref{thm:algorithm1}, 
it is easy to verify that the \LRWD problem can be solved in time $2^{\mathcal{O}(k)}\cdot n^4$.
\end{proof}

\section{A lower bound for \LRWD}\label{sec:lowerbound}
In this section, we show that the algorithm in Theorem~\ref{thm:main1} is best possible in some sense. 
Our lower bound is based on a well-known complexity hypothesis formulated by Impagliazzo, Paturi, and Zane~\cite{ImpagliazzoRF2001}.

\medskip
\medskip
	\textbf{Exponential Time Hypothesis (ETH).} There is a constant $s>0$ such that \textsc{3-CNF-SAT} with $n$ variables and $m$ clauses cannot be solved in time $2^{sn}(n + m)^{\mathcal{O}(1)}$.

\medskip
\medskip

We use the known lower bound for the \textsc{Vertex Cover} problem.

\smallskip
\noindent
\fbox{\parbox{0.97\textwidth}{
{\sc Vertex Cover} \\
\textbf{Input :} A graph $G$, a positive integer $k$ \\
\textbf{Parameter :} $k$ \\
\textbf{Question :} Does $G$ have a vertex subset $S$ of size at most $k$ such that $G\setminus S$ has no edges?} }

\begin{theorem}[Cai and Juedes~\cite{Cai2003}]\label{thm:vclow}
There is no $2^{o(k)}\cdot n^{\mathcal{O}(1)}$-time algorithm for \textsc{Vertex Cover}, unless ETH fails.
\end{theorem}

We show the following.

\begin{theorem}
There is no $2^{o(k)}\cdot n^{\mathcal{O}(1)}$-time algorithm for \textsc{LRW1 Vertex Deletion}, unless ETH fails.
\end{theorem}
\begin{proof}
For contradiction, suppose there exists an algorithm for solving the \LRWD problem in time $2^{o(k)}\cdot n^{\mathcal{O}(1)}$.
Let $(G, k)$ be an instance of the \textsc{Vertex Cover} problem.
We construct a graph $G'$ from $G$ as follows:
\begin{enumerate}[(1)]
\item for every vertex $v\in V(G)$, add a pendant vertex $v'$ adjacent to $v$, and
\item  for every edge $vw$ in $G$, we replace it with two vertex disjoint paths of length 2 from $v$ to $w$.
\end{enumerate}
Let $G'$ be the resulting graph. 
Note that for each edge $vw$ in $G$, 
$v, v', w, w'$ and two disjoint paths of length $2$ from $v$ to $w$ in $G'$ form an induced subgraph isomorphic to $\alpha_1$.
We have $\abs{V(G')}\le 2\abs{V(G)}+2\abs{E(G)}$.

We claim that $G$ has a vertex set $S$ of size at most $k$ such that $G\setminus S$ has no edges  if and only if $G'$ has a \DEL\ of size at most $k$.
Suppose that $G$ has a vertex set $S$ of size at most $k$ such that $G\setminus S$ has no edges.
It is easy to confirm that $G'\setminus S$ is a disjoint union of stars, which has linear rankwidth at most $1$.

For the converse direction, suppose that $G'$ has a \DEL\ $S$ of size at most $k$. 
If $S$ contains a vertex of degree $1$, then we can replace it with its neighbor.
We may assume that $S$ has no pendant vertices of $G'$.
Let $vw$ be an edge of $G$ and $x^{vw}_1, x^{vw}_2$ be the vertices of degree $2$ that are adjacent to $v$ and $w$ in $G'$.
As $S$ does not contain pendant vertices of $G'$, we have $\abs{S\cap \{v, w, x^{vw}_1, x^{vw}_2\}}\ge 1$, otherwise, $G'$ contains an induced subgraph isomorphic to $\alpha_1$.
For each edge $vw$, if $S$ contains one of $x^{vw}_1, x^{vw}_2$, then we replace it with one of $v$ and $w$. Let $S'$ be the resulting set.
Then $\abs{S'\cap V(G)}\le \abs{S}\le k$ and $S'\cap V(G)$ contains at least one of $v$ and $w$ for each edge $vw$ of $G$.
We conclude that $G$ has a vertex set $S$ of size at most $k$ such that $G\setminus S$ has no edges.

Therefore, using the algorithm for the \LRWD problem, we can decide whether $(G, k)$ is a \YES-instance of the \textsc{Vertex Cover} problem, in time $2^{o(k)}\cdot n^{\mathcal{O}(1)}$, 
which is not possible unless ETH fails by Theorem~\ref{thm:vclow}.
\end{proof}

\section{A polynomial kernel for \textsc{LRW1-Vertex Deletion}}\label{sec:polykerthreaddel}

In this section, we show the following.
\begin{theorem}\label{thm:main22}
The \LRWD problem has a kernel with $\mathcal{O}(k^{33})$ vertices.
\end{theorem}

We use the Sunflower lemma to find a hitting set for obstructions of small size with a special property. It consists in finding a subset $T$ with size bounded by $f(k)$ for some function $f$ whose removal turns $G$ into a graphs of linear rankwidth at most $1$ with the property that  for every set $S\subseteq V(G)$ of size at most $k$, the following are equivalent (Lemma~\ref{lem:shrink}):
\begin{itemize}
\item $S$ is a minimal vertex set such that $G\setminus S$ has no obstructions in $\obn$. 
\item $S$ is a minimal vertex set such that $G[T]\setminus S$ has no obstructions in $\obn$.
\end{itemize} 
This property implies that if there is a minimal LRW1-deletion set $S$, then each vertex of $S\setminus T$ should be used to remove at least one long induced cycle.
It will be used to find an irrelevant vertex in a large thread block in $G\setminus T$.
Moreover, we can preprocess the instance so that 
there is no small obstruction containing exactly one vertex of $T$. 
This will be used to bound the length of the sequence of thread blocks in each connected component.

Let $(G,k)$ be an instance of {\sc LRW$1$-Vertex Deletion}. 
We start with an easy reduction rule.

\begin{RULE}\label{rule:threadcomponent}
If $G$ has a connected component that has linear rankwidth at most $1$, 
then we remove it from $G$.
\end{RULE}
\subsection{Hitting small obstructions}

Let $\cF$ be a family of subsets over a set $U$.  A subset $U'\subseteq U$ is called \emph{a hitting set} of $\cF$ if for every set $F\in \cF$, $F\cap U'\neq\emptyset$. 
For a graph $G$ and a family of
graphs $\cF$, a set $S\subseteq V(G)$ is also called \emph{a hitting set} for $\cF$ if for every induced subgraph $H$ of $G$ that is isomorphic to a graph in $\cF$, $V(H)\cap S\neq\emptyset$. 
The following lemma can be obtained from the Sunflower lemma.
	 
\begin{lemma}[Fomin, Saurabh, and Villanger~\cite{FSV2012}]\label{lem:fomin}
Let $\cF$ be a family of sets of size at most $d$ over a set $U$, and let $k$ be a positive integer. Then in time $\mathcal{O}(\abs{\cF}(k + \abs{\cF}))$, we can find a nonempty set $\cF'\subseteq \cF$ such that
\begin{enumerate}
\item  for every $U'\subseteq U$ of size at most $k$, $U'$ is a minimal hitting set of $\cF$ if and only if $U'$ is a minimal hitting set of $\cF'$, and
\item $\abs{\cF'} \le d!(k+1)^d$.
\end{enumerate}
\end{lemma}

Using Proposition~\ref{prop:nltothread2} and Lemma~\ref{lem:fomin}, we identify a vertex set $T$ of $G$ with size polynomial in $k$ that allows us to forget about small obstructions. 

\begin{lemma}\label{lem:shrink}
Let $(G, k)$ be an instance of {\sc LRW$1$-Vertex Deletion}. There is a polynomial time algorithm that either concludes that $(G, k)$ is a \NO-instance or finds a non-empty set $T\subseteq V(G)$ such that
\begin{enumerate}
\item $G\setminus T$ has linear rankwidth at most $1$,
\item for every set $S\subseteq V(G)$ of size at most $k$, $S$ is a minimal hitting set for $\obn$ in $G$ if and only if it is a minimal hitting set for $\obn$ contained in $G[T]$, and
\item $\abs{T}\le 8\cdot 8!(k + 1)^8 + k$.
\end{enumerate}
\end{lemma}

\begin{proof}
Let $\cF$ be the set of vertex sets $S$ of $G$ such that $G[S]$ is isomorphic to a graph in $\obn$. 
Since the maximum size of a set in $\mathcal{F}$ is $8$, 
using Lemma~\ref{lem:fomin}, we can find a subset $\mathcal{F}'$ of $\mathcal{F}$ such that
\begin{enumerate}
\item   for every vertex subset $X\subseteq V(G)$ of size at most $k$, 
	$X$ is a minimal hitting set of $\cF$ if and only if $X$ is a minimal hitting set of $\cF'$, and
\item  $\abs{\cF'} \le 8!(k+1)^8$.
\end{enumerate}

Let $T':=\bigcup_{S\in \mathcal{F}'} S$. 
From the condition 1,
$G\setminus T'$ has no induced subgraph isomorphic to a graph in $\obn$ and by Theorem~\ref{thm:mainlrw}, $G\setminus T'$ is a necklace graph.
Using the algorithm in Proposition~\ref{prop:nltothread2}, we can find a minimum LRW1-deletion set $Y$ of $G\setminus T'$ in polynomial time.
If $\abs{Y}\ge k+1$, then we conclude that $(G, k)$ is a \NO-instance.
Otherwise, we add $Y$ to $T'$, increasing its size by at most $k$.
We conclude that $T:=T'\cup Y$ is a required set.
\end{proof}

Let us fix a subset $T$ of $V(G)$ obtained by Lemma~\ref{lem:shrink}.  
We preprocess using the following reduction rule.

\begin{RULE}\label{rule:onevertex}
Let $U\subseteq T$ such that for every $u\in U$, there exists an induced subgraph $H$ of $G$ isomorphic to a graph in $\obn$ with $V(H)\cap T=\{u\}$.
If $\abs{U} > k$, then $(G,k)$ is a \NO-instance; otherwise, remove $U$ from $G$ and reduce $k$ by $\abs{U}$, and use $T\setminus U$ instead of $T$.
\end{RULE}
It can be done in polynomial time because we only need to look at obstructions of $\obn$ in $G$. 

\begin{lemma}\label{lem:oneintersect}
 Reduction Rule~\ref{rule:onevertex} is safe.
\end{lemma}
\begin{proof}
We claim that every minimal LRW1-deletion set in $G$ contains $U$.
Let $S$ be a minimal LRW1-deletion set in $G$. Then there exists a vertex subset $S'\subseteq S\cap T$ such that $S'$ is a minimal hitting set for graphs of $\obn$ in $G[T]$. 
From the property of $T$, $S'$ is also a minimal hitting set for graphs of $\obn$ in $G$, and we must have $U\subseteq S'$ as $S'$ hits the sets $S_u$ for each $u\in U$.
It implies that if $\abs{U}>k$, then $(G,k)$ is a \NO-instance.
Otherwise, since $U$ is always contained in any minimal LRW1-deletion set of $G$, we have that  $(G, k)$ is a \textsc{Yes}-instance if and only if $(G\setminus U, k-\abs{U})$ is a \textsc{Yes}-instance.
\end{proof}

From now on, we assume that $G$ is reduced under Reduction Rules~\ref{rule:threadcomponent} and \ref{rule:onevertex}.
A  vertex $v$ of $G$ is called \emph{irrelevant} if $(G,k)$ is a \YES-instance if and only if $(G\setminus v, k)$ is a \YES-instance. 
For convenience, let $\mu(k):=8\cdot 8!(k+1)^8 +k$.

\subsection{Bound on the size of connected components of $G\setminus T$}\label{subsec:boundcomp}

We first show that if a thread block in
$G\setminus T$ is large, then we can always find an irrelevant vertex in there.

\begin{proposition}\label{prop:irrelevant} If $G\setminus T$ contains a thread block of size at least $(k+2)(\mu(k)+2)^2+1$, then we can find an irrelevant vertex
  in polynomial time.
\end{proposition}

We use the following lemma.

\begin{lemma}\label{lem:diamond}
Let $G$ be a graph and let $v_1v_2v_3v_4v_5$ be an induced path of length $4$ in $G$.
If two distinct vertices $w_1, w_2$ in $V(G)\setminus \{v_1, v_2, \ldots, v_5\}$ are adjacent to $v_2$ and $v_4$,
then $G\setminus v_3$ contains an induced subgraph isomorphic to a graph in $\obn$. 
\end{lemma}
\begin{proof}
See Figure~\ref{fig:lemma65} for the following cases.
If $v_1$ is adjacent to $w_1$ but not adjacent to $w_2$, then $v_1v_2w_2v_4$ is an induced path of length $3$ and $w_1$ is adjacent to its end vertices.
By Lemma~\ref{lem:pathtoobs}, $G[\{v_1, v_2, w_2, v_4, w_1\}]$ has an induced subgraph isomorphic to a graph in $\obn$.
Considering all symmetric cases, we may assume that 
for each $v\in \{v_1, v_5\}$, 
$v$ is adjacent to both $w_1, w_2$ or neither of them.
Depending on the adjacency between $\{v_1, v_5\}$ and $\{w_1, w_2\}$, and the adjacency between $w_1$ and $w_2$, we have one of the $6$ graphs in $\obn$, which are $\alpha_1, \alpha_2, \ldots, \alpha_6$.
\end{proof}

\begin{figure}
\centerline{\includegraphics[scale=0.8]{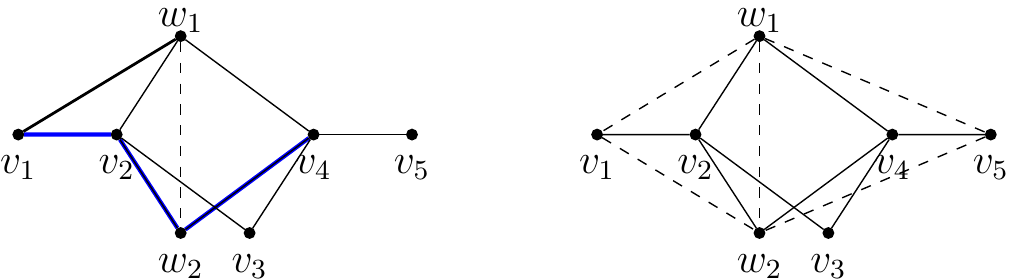}}
\caption{Two cases in Lemma~\ref{lem:diamond}.  }\label{fig:lemma65}
\end{figure}

\begin{proof}[Proof of Proposition~\ref{prop:irrelevant}]
Suppose that $G\setminus T$ contains a thread block of size at least $(k+2)(\mu(k)+2)^2+1$.
We can find such a thread block $B$ and compute its ordering $\sigma$ and labeling $\ell$ in polynomial time using the algorithm in Lemma~\ref{lem:splittreetothreadblock}.
 Let $x,y$ be the first and last vertices of $B$ and let $\sigma'$ be the ordering obtained from $\sigma$ by removing $x$ and $y$.

We mark some vertices of $B$ as follows. We set $Z:=\emptyset$. 
\begin{enumerate}[(1)]
\item For each vertex $v$ of $T$, choose the first $k+2$ vertices $z$ of $\sigma'$ that are neighbors of $v$ with $R\in \ell(z)$, and choose the last $k+2$ vertices $z$ of $\sigma'$ that are neighbors of $v$ with $L\in \ell(z)$, and add them to $Z$. 
\item For each pair of two vertices $v$, $v'$ in $T$, choose $k+2$ common neighbors of $v$ and $v'$ in $B$, and add them to $Z$. 
\item  Choose the first $k+2$ vertices $z$ of $\sigma'$ with $R\in \ell(z)$, and choose the last $k+2$ vertices $z$ of $\sigma'$ with $L\in \ell(z)$, and add them to $Z$.
\end{enumerate}
In each case, if there are at most $k+1$ such vertices, then we add all of them to $Z$.
Then
\begin{align*} 
\abs{Z}\le \abs{T}(2k+4) + \abs{T}^2(k+2) + (2k+4) \le (k+2)(\mu(k) +2)^2 -2.\end{align*}
Since $\abs{V(B)}\ge (k+2)(\mu(k) +2)^2+1$, 
there exists a vertex $w$ in $V(B)\setminus (Z\cup \{x,y\})$.

We claim that $w$ is an irrelevant vertex. 
 If $(G, k)$ is a \YES-instance, then $(G\setminus w, k)$ is clearly a \YES-instance.

Suppose that there is a vertex set $X\subseteq V(G\setminus w)$ with $\abs{X}\le k$ such that $G\setminus (X\cup \{w\})$ is a thread graph.
We may assume that $G\setminus X$ is not a thread graph.
So, $G\setminus X$ has an induced subgraph containing $w$ that is isomorphic to a graph in $\obt$. 
Let $X'\subseteq X\cup \{w\}$ be a minimal hitting set for $\obn$ in $G$.
From the property of the set $T$, $X'$ is a minimal hitting set for $\obn$ in $G[T]$, which implies that $X'\subseteq T$.
Thus $G\setminus X$ must have an induced cycle of length at least $9$ that contains $w$.
Let $C$ be an induced cycle of length at least $9$ containing $w$ in $G\setminus X$.

We will find an induced subgraph of $G\setminus (X\cup \{w\})$ that is isomorphic to a graph in $\obn$, which leads to a contradiction.
Let $v_1, v_2, w, v_3, v_4$ be the consecutive vertices on $C$.
To apply Lemma~\ref{lem:diamond}, we find two vertices that are adjacent to $v_2$ and $v_3$.
These cases are depicted in Figure~\ref{fig:irrelevant}.

\begin{figure}
\centering{\includegraphics[scale=0.8]{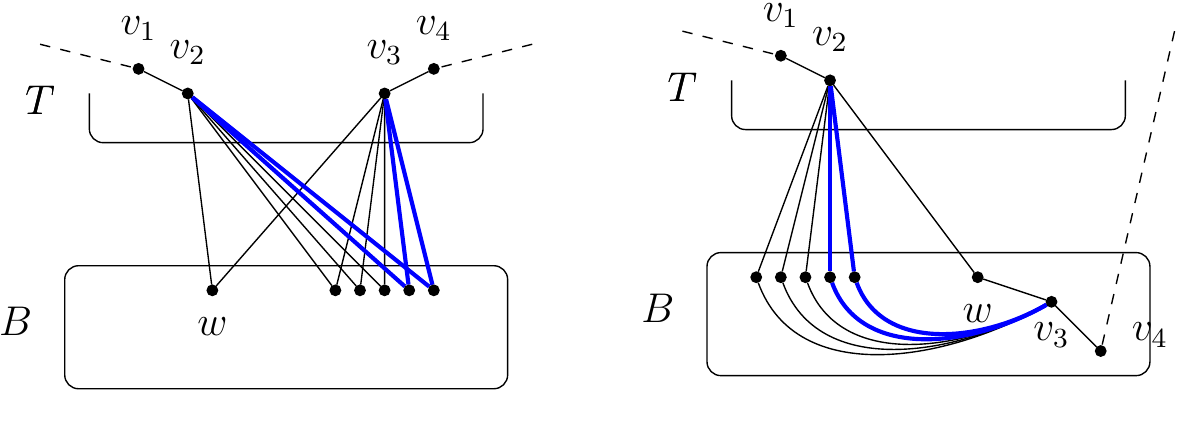}
\,\,\quad \includegraphics[scale=0.8]{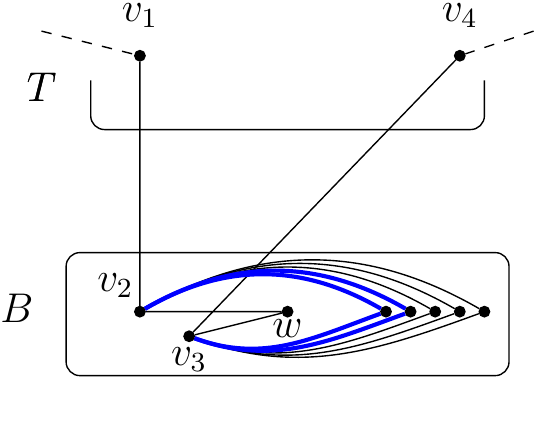}}
\caption{Cases 1-3 in Proposition~\ref{prop:irrelevant}.} \label{fig:irrelevant}
\end{figure}

\begin{enumerate}
\item (Case 1. $v_2, v_3\in T$.) Since $v_2$ and $v_3$ have a common neighbor $w$ in $V(B)\setminus Z$, 
$Z$ contains at least $k+2$ common neighbors of $v_2$ and $v_3$.
Since $\abs{X}\le k$, there exist two vertices $w_1, w_2\in Z\setminus X$ 
that are common neighbors of $v_2$ and $v_3$.
\item (Case 2. One of $v_2$ and $v_3$ is contained in $T$.) From the symmetry, we may assume that $v_2\in T$ and $v_3\notin T$.
Since $w\notin \{x,y\}$,
$v_3$ is contained in $B$. If $R\in \ell(w)$ and $w<_{\sigma} v_3$, then $Z$ contains the first $k+2$ vertices $z$
of $\sigma'$ that are neighbors of $v_2$ with $R\in \ell(z)$. We choose two vertices of them that are not in $X$.
In case when $L\in \ell(w)$ and $v_3<_\sigma w$, we use the last $k+2$ vertices $z$ of $\sigma'$ that are neighbors of $v_2$ with $L\in \ell(z)$ to identify two vertices similarly.
\item (Case 3. Neither $v_2$ nor $v_3$ is contained in $T$.)
Since $w\notin \{x,y\}$, $v_2$ and $v_3$ are contained in $B$.
If $v_2<_{\sigma} w<_{\sigma} v_3$, 
then $R\in \ell(v_2)$, $L\in \ell(v_3)$ and it implies that $v_2v_3\in E$, which is a contradiction.
Also, $v_3<_{\sigma} w<_{\sigma} v_2$ cannot happen.
Thus, both of $v_2$ and $v_3$ appear either before $w$ in $\sigma$ or after $w$ in $\sigma$.
By the symmetry, we may assume that $v_2$ and $v_3$ appear before $w$ in $\sigma$.
So, $R\in \ell(v_2)$, $R\in \ell(v_3)$, and $L\in \ell(w)$.
Since $Z$ contains the last $k+2$ vertices $z$ of $\sigma'$ with $L\in \ell(z)$, there exist two vertices $w_1, w_2$ from those $k+2$ vertices that are not in $X$ and $C$.\end{enumerate}

In all cases, $G\setminus (X\cup \{w\})$ has an induced subgraph isomorphic to a graph in $\obn$ by Lemma~\ref{lem:diamond}. It contradicts to the assumption that $X\cup \{w\}$ is a LRW1-deletion set of $G$.
Therefore, $G\setminus X$ is a thread graph, and
we conclude that $(G,k)$ is a \YES-instance.
\end{proof}

In the next lemma, we show that if a vertex $v$ in $T$ has neighbors on $7$ distinct blocks in a connected component of $G\setminus T$, 
then we can find a subgraph $H$ isomorphic to one of $\{\beta_1,\beta_2,\beta_3,\beta_4\}$ such that $V(H)\cap T=\{v\}$, which is not possible by Reduction Rule~\ref{rule:onevertex}. 
Using it,  whenever a connected component of $G\setminus T$ has a long sequence of thread blocks, we can identify a sequence of consecutive thread blocks not touched by any
obstruction in $\obn$. 
This allows us to contract one of these thread blocks to a vertex.

\begin{lemma}\label{lem:reducelength}
If $G\setminus T$ has a connected component with at least $19(6\mu(k)+1)$ thread blocks,
then we can in polynomial time transform $G$ into a graph $G'$ with $\abs{V(G')}< \abs{V(G)}$ such that $(G,k)$ is a \textsc{Yes}-instance if and only if $(G',k)$ is a \textsc{Yes}-instance.
\end{lemma}

\begin{proof}
Suppose that $G\setminus T$ has a connected component $H$ such that $H$ consists of at least $19(6\mu(k) +1)$ thread blocks.
We can find such a connected component in polynomial time  using the algorithm in Lemma~\ref{lem:splittreetothreadblock}, if exists. 
Let $B_1, B_2, \ldots, B_t$ be the sequence of thread blocks of $H$.

We claim that every vertex $v$ of $T$ has neighbors in at most $6$ thread blocks of $H$. 
For contradiction, suppose that there is a vertex $v$ in $T$ having neighbors in at least $7$ thread blocks.
Then there are three thread blocks $B_{t_1}, B_{t_2}, B_{t_3}$ having a neighbor of $v$ in $G$ such that  
$t_2-t_1\ge 3$ and $t_3-t_2\ge 3$.
For each $i\in \{1,2,3\}$, let $p_i$ be a neighbor of $v$ in $B_{t_i}$.
Since each thread block consists of at least two vertices, 
we can choose a neighbor $q_i$ of $p_i$ in $B_{t_i}$ for each $i\in \{1,2,3\}$.
Depending on the adjacency between $v$ and the vertices $q_1, q_2, q_3$, 
we have an induced subgraph  isomorphic to a graph in $\{\beta_1, \beta_2, \beta_3, \beta_4\}$ such that it has exactly one vertex of $T$.
This contradicts to the assumption that 
$(G,k)$ is an instance %
reduced by Reduction Rule~\ref{rule:onevertex}.

Now, for each vertex $v$ of $T$, we mark the thread blocks $B$ of $H$ containing a neighbor of $v$.
Since the number of thread blocks in $H$ is at least $10(6\mu(k) +1)$ and
$10(6\mu(k) +1)-6\mu(k) \ge 9(6\mu(k)+1)$,
there exist $m\ge 9$ and $1\le j\le t-m$ such that
$B_{j+1}, B_{j+2}, \ldots, B_{j+m}$ are non-marked thread blocks.

Let $x, y$ be the two end vertices of $B_{j+5}$.
We transform the graph $G$ into a graph $G'$ by removing the thread block $V(B_{j+5})$ and adding a new vertex $z$ that are adjacent to $(N_G(x)\cup N_G(y))\setminus V(B_{j+5})$. 
Let $H'$ be the connected component of $G'\setminus T$ that is modified from the connected component $H$ of $G\setminus T$.
Since we remove at least two vertices from $G$ and add one vertex, we have $\abs{V(G')}< \abs{V(G)}$.

We show that $(G,k)$ is a \textsc{Yes}-instance if and only if $(G', k)$ is a \textsc{Yes}-instance.
Suppose that $G$ has a minimal thread vertex set $X$.
We first assume that $X$ contains a vertex $q$ in $V(B_{j+5})$.  
Since $X$ is a minimal LRW1-deletion set and all small obstructions of $\obn$ are contained in $G\setminus V(B_{j+5})$, $q$ hits an induced cycle of length at least $9$ in $G$, and 
the cycle must pass through the vertices $x$ and $y$.
Thus, $(X\setminus V(B_{j+5}))\cup \{z\}$ is a LRW1-deletion set of $G'$ with size at most $\abs{X}$.

Assume that $X\cap V(B_{j+5})= \emptyset$. 
We may assume that $G'\setminus X$ is not a thread graph. 
Then $G'\setminus X$ must have an induced cycle $C$ of length at least $9$ intersecting the new vertex $z$.
The cycle obtained from $C$ by replacing $z$ with the edge $xy$ is also an induced cycle of length at least $9$ in $G\setminus X$. It contradicts to the assumption that 
$G\setminus X$ is a thread graph.

Now suppose that $G'$ has a minimal LRW1-deletion set $X$.
If $z\in X$, then $z$ hits an induced cycle of length at least $9$ in $G'$ because of the minimality of $X$ and the distance from $x$ to the end vertices of $B_{j+1}$ and $B_{j+m}$.
Because $x$ hits all induced cycles of length at least $9$ in $G$ having a vertex of $V(B_{j+5})$, $(X\setminus \{z\})\cup \{x\}$ is again a LRW1-deletion set of $G$. 

Assume that $z\notin X$. 
We may assume that $G\setminus X$ is not a thread graph.
So, $G\setminus X$ has an induced cycle $C$ of length at least $9$ passing through $x$ and $y$.
Let $C'$ be the cycle obtained from $C$ by replacing the edge $xy$ with the vertex $z$.
This cycle $C'$ clearly exists in $G'\setminus X$ and it has length at least $9$ because it should contain at least one vertex from the thread blocks $B_{j+1}, \ldots, B_{j+4}, B_{j+6}, \ldots B_{j+m}$ with $m\ge 10$.
This contradicts to the assumption that $G'\setminus X$ is a thread graph.
We conclude that $(G, k)$ is a \textsc{Yes}-instance if and only if $(G', k)$ is a \textsc{Yes}-instance.
\end{proof}

\subsection{Kernel size}\label{subsec:kernelsize}

We bound the number of connected components using the following lemma.

\begin{lemma}\label{lem:reducecomponent1}
\begin{enumerate}[(1)]
\item The graph $G\setminus T$ has at most $2\mu(k)$ connected components containing at least two vertices.
\item If $G\setminus T$ has at least $\mu(k)^2\cdot (k+2)+1$ isolated vertices,
then we can find an irrelevant vertex in polynomial time.
\end{enumerate}
\end{lemma}

\begin{proof}
(1) By Reduction Rule~\ref{rule:threadcomponent}, 
each connected component $H$ of $G\setminus T$ contains a vertex that has a neighbor in $T$.
Let $\mathcal{C}$ be the set of connected components of $G\setminus T$ which consist of at least two vertices, and suppose that $\abs{\mathcal{C}}> 2\mu(k)$.
Since every connected component of $H$ has a vertex having a neighbor in $T$, 
there exists a vertex $u\in T$ such that $u$ has neighbors in three distinct connected components of $\mathcal{C}$.
Since each connected component of $\mathcal{C}$ has at least two vertices, 
$G$ has a vertex set $S$ where $G[S]$ is isomorphic to a graph in $\{\beta_1, \beta_2, \beta_3, \beta_4\}$ and $S\cap T=\{u\}$. It contradicts to the assumption that $(G,k)$ is reduced
by Reduction Rule~\ref{rule:onevertex}.

\medskip
(2) Suppose that $G\setminus T$ has at least $\mu(k)^2\cdot (k+2)+1$ isolated vertices.
Let $S$ be the union of isolated vertices in $G\setminus T$. 
We may assume that every vertex in $S$ has a neighbor in $T$. 
 
We define a set $Z$ to identify an irrelevant vertex.
For each pair of two vertices in $T$, choose $k+2$ common neighbors in $S$, and add them to $Z$. If there are at most $k+1$ common neighbors, then we add all of them into $Z$.
Since $\abs{S}>\mu(k)^2\cdot (k+2)$, there is a vertex $w$ in $S\setminus Z$.

 We claim that $w$ is an irrelevant vertex of the problem. 
 If $(G,k)$ is a \textsc{Yes}-instance,
 then there exists a vertex subset $X$ of size at most $k$ in $G$ such that $G\setminus X$ is a thread graph.
 Since $G\setminus (X\cup \{w\})$ is also a thread graph, $(G\setminus w, k)$ is a \textsc{Yes}-instance.
 
 Suppose that $(G\setminus w, k)$ is a \textsc{Yes}-instance.
 We choose a minimal vertex set $X$ in $G\setminus w$ such that $\abs{X}\le k$ and $G\setminus (X\cup \{w\})$ is a thread graph.
 We may assume that $G\setminus X$ is not a thread graph.
 Let $X'\subseteq X\cup \{w\}$ be a hitting set for $\obn$ in $G[T]$.
 Then by the property of $T$, $X'$ also hits all induced subgraphs in $G$ that are isomorphic to a graph of $\obn$.
 Since $X$ already hits all small obstructions in $G$,
 there exists an induced cycle $C$ of length at least $9$ in $G\setminus X$ containing $w$.

 Let $w_1, w_2$ be the neighbors of $w$ on the cycle $C$.
 Since $w_1, w_2$ have $k+2$ common neighbors in $Z$, 
 we may choose two vertices $z_1, z_2\in Z\setminus X$ that are common neighbors of $w_1$ and $w_2$.
 By Lemma~\ref{lem:diamond}, we have that $G\setminus (X\cup \{w\})$ has an induced subgraph isomorphic to a graph in $\obn$,
 which implies that $G\setminus (X\cup \{w\})$ is not a thread graph. 
 It is a contradiction, and we conclude that $(G, k)$ is a \textsc{Yes}-instance.
\end{proof}

Let us now piece everything together and analyze the kernel size.

\begin{proof}[Proof of Theorem~\ref{thm:main22}]
  Let $(G,k)$ be an instance of {\sc LRW1-Vertex Deletion}. 
  By Reduction Rule~\ref{rule:threadcomponent}, we may safely assume that $G$ has no connected components that are thread graphs. 
  Then using the algorithm in Lemma~\ref{lem:shrink}, in polynomial time, either we conclude that $(G, k)$ is a \NO-instance or find a non-empty set $T\subseteq V(G)$ stated in Lemma~\ref{lem:shrink}.
  We apply Reduction Rule~\ref{rule:onevertex}.
Lemma~\ref{lem:oneintersect} guarantees that for every vertex set $S\subseteq V(G)$ such that $G[S]$ is isomorphic to a graph in $\{\beta_1, \beta_2, \beta_3, \beta_4\}$, $\abs{S\cap T}\ge 2$. 

Combining Proposition~\ref{prop:irrelevant} and  Lemma~\ref{lem:reducelength}, we can assume that every connected component of $G\setminus T$ has size at most  $(k+2)(\mu(k)+2)^2\cdot 19(6\mu(k)+1)$ (otherwise the instance can be reduced in polynomial time). 
Note that for each connected component $H$ of $G\setminus T$, there exists a vertex in $H$ that has a neighbor in $T$. Therefore, by Lemma~\ref{lem:reducecomponent1}, the number of non-trivial components of $G\setminus T$ is at most $2\mu(k)$ and
the number of isolated vertices in $G\setminus T$ is at most $\mu(k)^2(k+2)$.
It follows that
\begin{align*}
\abs{T}+\abs{V(G)\setminus T}&\le 
\mu(k) + \left( 2\mu(k) \cdot 19(6\mu(k)+1) \cdot (k+2)(\mu(k)+2)^2 + \mu(k)^2\cdot (k+2) \right) \\
&=\mathcal{O}(k\cdot \mu(k)^4 ) =\mathcal{O}(k^{33}). \qedhere
\end{align*}
\end{proof}
\section{Concluding remarks}\label{sec:remark}

We consider the problem {\sc Linear rankwidth-$w$ Vertex Deletion} when $w=1$. A next step is to investigate the problem for bigger $w$, or for any fixed $w$. A closely related problem is {\sc Rankwidth-$w$ Vertex Deletion}, which asks whether $G$ has a vertex subset of size at most $k$ such that $G\setminus S$ has rankwidth at most $w$. {\sc (Linear) Rankwidth-$w$ Vertex Deletion} is fixed-parameter tractable for the following reason. Note that any \textsc{Yes}-instance has rankwidth at most $w+k$. Having bounded (linear) rankwith can be characterized by a finite list of forbidden vertex-minors~\cite{Oum05}. From~\cite{CO2007}, having a vertex-minor can be expressed in $\sf{C}_2\sf{MSO}$, i.e., monadic second order logic without edge set quantification where we can express the parity of $|X|$ for a vertex set $X$. Fixed-parameter tractability follows as a consequence of Courcelle, Makowsky, Rotics~\cite{CourcelleMR00}. 

As for rankwidth, this result can be turned into a constructive algorithm as~\cite{Oum05} provides an explicit upper bound on the size of vertex-minor obstructions for rankwidth at most $k$ for fixed $k$. 
However, the exponential blow-up in the running time is huge with respect to both $w$ and $k$. It is a challenging question whether a reasonable dependency on $k$ can be achieved. A single-exponential time would be ideal, which was achievable for its treewidth counterpart. A first realistic goal is to consider the case when $w=1$, that is, the {\sc Distance-Hereditary Vertex Deletion}. We leave it as an open question whether this problem can be solved in time  $c^k\cdot {n}^{\mathcal{O}(1)}$ time for some constant $c$. For linear rankwidth, there is no known upper bound on the size of vertex-minor obstructions for linear rankwidth at most $k$, and thus, obtaining such an upper bound is an interesting open question. 

\section*{Acknowledgment}
The third author would like to thank Sang-il Oum for suggesting the refined branching algorithm using cliquewidth.


\begin{thebibliography}{10}

\bibitem{AFP2013}
I.~Adler, A.~M. Farley, and A.~Proskurowski.
\newblock Obstructions for linear rank-width at most 1.
\newblock {\em Discrete Applied Mathematics}, 168:3--13, 2014.

\bibitem{AKK2014}
I.~Adler, M.~M. Kant{\'{e}}, and O.~Kwon.
\newblock Linear rank-width of distance-hereditary graphs.
\newblock In {\em Graph-Theoretic Concepts in Computer Science - 40th
  International Workshop, {WG} 2014, Nouan-le-Fuzelier, France, June 25-27,
  2014. Revised Selected Papers}, pages 42--55, 2014.

\bibitem{AdlerKK20152}
I.~Adler, M.~M. Kant\'e, and O.~Kwon.
\newblock Linear rank-width of distance-hereditary graphs {II}. vertex-minor
  obstructions.
\newblock {\em preprint}, arxiv.org/abs/1508.04718, 2015.

\bibitem{BM1986}
H.-J. Bandelt and H.~M. Mulder.
\newblock Distance-hereditary graphs.
\newblock {\em J. Combin. Theory Ser. B}, 41(2):182--208, 1986.

\bibitem{Bui-XuanKL13}
B.~Bui{-}Xuan, M.~M. Kant{\'{e}}, and V.~Limouzy.
\newblock A note on graphs of linear rank-width 1.
\newblock {\em preprint}, arxiv.org/abs/1306.1345, 2013.

\bibitem{Cai2003}
L.~Cai and D.~Juedes.
\newblock On the existence of subexponential parameterized algorithms.
\newblock {\em Journal of Computer and System Sciences}, 67(4):789 -- 807,
  2003.
\newblock Parameterized Computation and Complexity 2003.

\bibitem{Cao2015}
Y.~Cao.
\newblock Unit interval editing is fixed-parameter tractable.
\newblock In M.~M. Halldórsson, K.~Iwama, N.~Kobayashi, and B.~Speckmann,
  editors, {\em Automata, Languages, and Programming}, volume 9134 of {\em
  Lecture Notes in Computer Science}, pages 306--317. Springer Berlin
  Heidelberg, 2015.

\bibitem{CaoM2015}
Y.~Cao and D.~Marx.
\newblock Interval deletion is fixed-parameter tractable.
\newblock {\em ACM Trans. Algorithms}, 11(3):Art. 21, 35, 2015.

\bibitem{ChvatalH1977}
V.~Chv{\'a}tal and P.~L. Hammer.
\newblock Aggregation of inequalities in integer programming.
\newblock In {\em Studies in integer programming ({P}roc. {W}orkshop, {B}onn,
  1975)}, pages 145--162. Ann. of Discrete Math., Vol. 1. North-Holland,
  Amsterdam, 1977.

\bibitem{Cou90}
B.~Courcelle.
\newblock {The Monadic Second-Order Theory of Graphs. {I}. {R}ecognizable Sets
  of Finite graphs}.
\newblock {\em Information and Computation}, 85:12--75, 1990.

\bibitem{CourcelleK09}
B.~Courcelle and M.~M. Kant{\'{e}}.
\newblock Graph operations characterizing rank-width.
\newblock {\em Discrete Applied Mathematics}, 157(4):627--640, 2009.

\bibitem{CourcelleMR00}
B.~Courcelle, J.~A. Makowsky, and U.~Rotics.
\newblock Linear time solvable optimization problems on graphs of bounded
  clique-width.
\newblock {\em Theory Comput. Syst.}, 33(2):125--150, 2000.

\bibitem{CO2007}
B.~Courcelle and S.~Oum.
\newblock Vertex-minors, monadic second-order logic, and a conjecture by
  {S}eese.
\newblock {\em J. Combin. Theory Ser. B}, 97(1):91--126, 2007.

\bibitem{Cunningham1982}
W.~H. Cunningham.
\newblock Decomposition of directed graphs.
\newblock {\em SIAM J. Algebraic Discrete Methods}, 3(2):214--228, 1982.

\bibitem{CyganPPW2012}
M.~Cygan, M.~Pilipczuk, M.~Pilipczuk, and J.~Wojtaszczyk.
\newblock An improved fpt algorithm and a quadratic kernel for pathwidth one
  vertex deletion.
\newblock {\em Algorithmica}, 64(1):170--188, 2012.

\bibitem{Dahlhaus00}
E.~Dahlhaus.
\newblock Parallel algorithms for hierarchical clustering, and applications to
  split decomposition and parity graph recognition.
\newblock {\em Journal of Algorithms}, 36(2):205--240, 2000.

\bibitem{FominLMS12}
F.~V. Fomin, D.~Lokshtanov, N.~Misra, and S.~Saurabh.
\newblock Planar f-deletion: Approximation, kernelization and optimal {FPT}
  algorithms.
\newblock In {\em 53rd Annual {IEEE} Symposium on Foundations of Computer
  Science, {FOCS} 2012, New Brunswick, NJ, USA, October 20-23, 2012}, pages
  470--479, 2012.

\bibitem{FSV2012}
F.~V. Fomin, S.~Saurabh, and Y.~Villanger.
\newblock A polynomial kernel for proper interval vertex deletion.
\newblock {\em SIAM Journal on Discrete Mathematics}, 27(4):1964--1976, 2013.

\bibitem{FG04}
M.~Frick and M.~Grohe.
\newblock The complexity of first-order and monadic second-order logic
  revisited.
\newblock {\em Ann. Pure Appl. Logic}, 130(1-3):3--31, 2004.

\bibitem{Ganian10}
R.~Ganian.
\newblock Thread graphs, linear rank-width and their algorithmic applications.
\newblock In {\em Combinatorial algorithms}, volume 6460 of {\em Lecture Notes
  in Comput. Sci.}, pages 38--42. Springer, Heidelberg, 2011.

\bibitem{GanianH10}
R.~Ganian and P.~Hlinen{\'{y}}.
\newblock On parse trees and myhill-nerode-type tools for handling graphs of
  bounded rank-width.
\newblock {\em Discrete Applied Mathematics}, 158(7):851--867, 2010.

\bibitem{GaspersS12}
S.~Gaspers and S.~Szeider.
\newblock Backdoors to satisfaction.
\newblock In {\em The Multivariate Algorithmic Revolution and Beyond - Essays
  Dedicated to Michael R. Fellows on the Occasion of His 60th Birthday}, pages
  287--317, 2012.

\bibitem{GeelenGW2006}
J.~Geelen, B.~Gerards, and G.~Whittle.
\newblock On {R}ota's conjecture and excluded minors containing large
  projective geometries.
\newblock {\em J. Combin. Theory Ser. B}, 96(3):405--425, 2006.

\bibitem{GP2012}
E.~Gioan and C.~Paul.
\newblock Split decomposition and graph-labelled trees: characterizations and
  fully dynamic algorithms for totally decomposable graphs.
\newblock {\em Discrete Appl. Math.}, 160(6):708--733, 2012.

\bibitem{HallOS2007}
R.~Hall, J.~Oxley, and C.~Semple.
\newblock The structure of 3-connected matroids of path width three.
\newblock {\em European J. Combin.}, 28(3):964--989, 2007.

\bibitem{ImpagliazzoRF2001}
R.~Impagliazzo, R.~Paturi, and F.~Zane.
\newblock Which problems have strongly exponential complexity?
\newblock {\em Journal of Computer and System Sciences}, 63(4):512 -- 530,
  2001.

\bibitem{JeongKO2016}
J.~Jeong, E.~J. Kim, and S.~Oum.
\newblock Constructive algorithm for path-width of matroids.
\newblock In R.~Krauthgamer, editor, {\em Proceedings of the Twenty-Seventh
  Annual {ACM-SIAM} Symposium on Discrete Algorithms, {SODA} 2016, Arlington,
  VA, USA, January 10-12, 2016}, pages 1695--1704. {SIAM}, 2016.

\bibitem{JKO2014}
J.~Jeong, O.~Kwon, and S.~Oum.
\newblock Excluded vertex-minors for graphs of linear rank-width at most {$k$}.
\newblock {\em European J. Combin.}, 41:242--257, 2014.

\bibitem{Kante2012}
M.~M. Kant{\'e}.
\newblock Well-quasi-ordering of matrices under {S}chur complement and
  applications to directed graphs.
\newblock {\em European J. Combin.}, 33(8):1820--1841, 2012.

\bibitem{KanteKKP2015}
M.~M. Kant{\'e}, E.~J. Kim, O.~Kwon, and C.~Paul.
\newblock {An FPT Algorithm and a Polynomial Kernel for Linear Rankwidth-1
  Vertex Deletion}.
\newblock In T.~Husfeldt and I.~Kanj, editors, {\em 10th International
  Symposium on Parameterized and Exact Computation (IPEC 2015)}, volume~43 of
  {\em Leibniz International Proceedings in Informatics (LIPIcs)}, pages
  138--150, Dagstuhl, Germany, 2015. Schloss Dagstuhl--Leibniz-Zentrum fuer
  Informatik.

\bibitem{Kashyap08}
N.~Kashyap.
\newblock Matroid pathwidth and code trellis complexity.
\newblock {\em SIAM J. Discrete Math.}, 22(1):256--272, 2008.

\bibitem{KLPRRSS13}
E.~J. Kim, A.~Langer, C.~Paul, F.~Reidl, P.~Rossmanith, I.~Sau, and S.~Sikdar.
\newblock Linear kernels and single-exponential algorithms via protrusion
  decompositions.
\newblock In {\em Automata, Languages, and Programming - 40th International
  Colloquium, {ICALP} 2013, Riga, Latvia, July 8-12, 2013, Proceedings, Part
  {I}}, pages 613--624, 2013.

\bibitem{KoutsonasTY2014}
A.~Koutsonas, D.~M. Thilikos, and K.~Yamazaki.
\newblock Outerplanar obstructions for matroid pathwidth.
\newblock {\em Discrete Math.}, 315:95--101, 2014.

\bibitem{Oum05}
S.~Oum.
\newblock Rank-width and vertex-minors.
\newblock {\em J. Comb. Theory, Ser. {B}}, 95(1):79--100, 2005.

\bibitem{Oum2006}
S.~Oum.
\newblock Approximating rank-width and clique-width quickly.
\newblock {\em ACM Trans. Algorithms}, 5(1):Art. 10, 20, 2008.

\bibitem{OS2004}
S.~Oum and P.~Seymour.
\newblock Approximating clique-width and branch-width.
\newblock {\em J. Combin. Theory Ser. B}, 96(4):514--528, 2006.

\bibitem{PhilipRV2010}
G.~Philip, V.~Raman, and Y.~Villanger.
\newblock A quartic kernel for pathwidth-one vertex deletion.
\newblock In D.~Thilikos, editor, {\em Graph Theoretic Concepts in Computer
  Science}, volume 6410 of {\em Lecture Notes in Computer Science}, pages
  196--207. Springer Berlin Heidelberg, 2010.

\bibitem{RS2004}
N.~Robertson and P.~D. Seymour.
\newblock Graph minors. {XX}. {W}agner's conjecture.
\newblock {\em J. Combin. Theory Ser. B}, 92(2):325--357, 2004.

\bibitem{BevernKM2010}
R.~van Bevern, C.~Komusiewicz, H.~Moser, and R.~Niedermeier.
\newblock Measuring indifference: unit interval vertex deletion.
\newblock In {\em Graph-theoretic concepts in computer science}, volume 6410 of
  {\em Lecture Notes in Comput. Sci.}, pages 232--243. Springer, Berlin, 2010.

\bibitem{VV2013}
P.~van't Hof and Y.~Villanger.
\newblock Proper interval vertex deletion.
\newblock {\em Algorithmica}, 65(4):845--867, 2013.

\end{thebibliography}
\end{document}